\documentclass[11pt]{article}

\usepackage{natbib}
\usepackage{booktabs}
\usepackage{multirow}
\usepackage{graphicx}
\usepackage{amsfonts}
\usepackage{amsmath}
\usepackage{amsthm}
\usepackage{vmargin}
\usepackage{times}
\usepackage{float}
\usepackage[colorlinks]{hyperref} \hypersetup{colorlinks=true, citecolor=red, breaklinks=true}
\usepackage[font={it}]{caption}

\newtheorem{proposition}{Proposition}

\newtheorem{hypothesis}{Hypothesis}

\DeclareMathOperator*{\argmax}{arg\,max}

\begin{document}

\title{Inducing Honest Reporting Without Observing Outcomes: An Application to the Peer-Review Process}

\author{
        Arthur Carvalho \\
        University of Waterloo\\
        a3carval@uwaterloo.ca
 \and
        Stanko Dimitrov\\
        University of Waterloo\\
        sdimitro@uwaterloo.ca
 \and
        Kate Larson \\
        University of Waterloo\\
        klarson@uwaterloo.ca
}

\date{\today}

\maketitle

\begin{abstract}
When eliciting opinions from a group of experts, traditional devices used to promote honest reporting assume that there is an observable future outcome. In practice, however, this assumption is not always reasonable. In this paper, we propose a scoring method built on strictly proper scoring rules to induce honest reporting without assuming observable outcomes. Our method provides scores based on pairwise comparisons between the reports made by each pair of experts in the group. For ease of exposition, we introduce our scoring method  by illustrating its application to the peer-review process. In order to do so, we start by modeling the peer-review process using a Bayesian model where the uncertainty regarding the quality of the manuscript is taken into account. Thereafter, we introduce our scoring method to evaluate the reported reviews. Under the assumptions that reviewers are Bayesian decision-makers and that they cannot influence the reviews of other reviewers, we show that risk-neutral reviewers strictly maximize their expected scores by honestly disclosing their reviews. We also show how the group's scores can be used to find a consensual review. Experimental results show that encouraging honest reporting through the proposed scoring method creates more accurate reviews than the traditional peer-review process.
\end{abstract}

\section{Introduction}
\label{sec:introduction}

In the absence of a well-chosen incentive structure, experts are not necessarily honest when reporting their opinions. For example, when reporting subjective probabilities,  experts who have a reputation to protect might tend to produce forecasts near the most likely group consensus, whereas experts who have a reputation to build might tend to overstate the probabilities of outcomes they feel will be understated in a possible consensus~\citep{Nakazono:2013}. Hence, an important question when eliciting experts' opinions is how to incentivize honest reporting.

\emph{Proper scoring rules}  \citep{Winkler:1968} are traditional devices that incentivize honest reporting of subjective probabilities, \textit{i.e.}, experts maximize their expected scores by honestly reporting their opinions. However, proper scoring rules rely on the assumption that there is an observable future outcome, which is not always a reasonable assumption. For example, when market analysts provide sales forecasts on a potential new product, there is no guarantee that the product will ever be produced. Hence, the actual number of sales may never be observed.

In this paper, we propose a scoring method for promoting honest reporting amongst a group of experts when future outcomes are unobservable. In particular, we are interested in settings where experts observe signals from a multinomial distribution with an unknown parameter. Honest reporting then means that experts report exactly the signals that they observed. Our scoring method is built on proper scoring rules. However, different than what is traditionally assumed in the proper scoring rules literature, our method does not assume that there is an observable future outcome. Instead, scores are determined based on pairwise comparisons between experts' reported opinions.

The proposed method may be used in a variety of settings, \textit{e.g.}, strategic planning, reputation systems, peer review, \textit{etc}.  When applied to strategic planning, the proposed method may induce honest evaluation of different strategic plans. A strategic plan is a systematic and coordinated way to develop a direction and a course for an organization, which includes a plan to allocate the organization's resources \citep{Argenti:1968}. After a candidate strategic plan is discarded, it becomes nearly impossible to observe what would be the consequences of that plan because strategic plans are long-term in nature. Hence, a method to incentivize honest evaluations of candidate strategic plans cannot assume that the result of a strategic plan is observable in the future.

Our method can also be applied to reputation systems to elicit honest feedback. In reputation systems, individuals rate a product/service after experiencing it, \textit{e.g.}, customer product reviews on Amazon.com are one such reputation system. Due to the subjective nature of this task, incentives for honest feedback should not be based on the assumption that an absolute rating exists.

For ease of exposition, we introduce our scoring method  by illustrating its application to a domain where traditionally there are no observable outcomes: the peer-review process. Peer review is a process in which an expert's output is scrutinized by a number of other experts with relevant expertise in order to ensure quality control and/or to provide credibility. Peer review is commonly used when there is no objective way to measure the output's quality, \textit{i.e.}, when quality is a subjective matter. Peer review has been widely used in several professional fields, \textit{e.g.}, accounting \citep{AICPA:2012}, law \citep{LSC:2005}, health care \citep{Dans:1993}, \textit{etc}.

Currently, a popular application of the peer-review process is in online education. Recent years have seen a surge of massive online open courses, \textit{i.e.}, free online academic courses aimed at large-scale participation. Some of these courses have attracted tens of thousands of students \citep{Pappano:2012}. One of the biggest challenges faced by online educators brought by this massive number of students is the grading process since the available resources (personnel, time, \textit{etc}.) is often insufficient. Auto-grading by computers is not always feasible, \textit{e.g.}, courses whose assignments consist of essay-style questions and/or questions that do not have clear right/wrong answers. Peer review has been used by some companies like Coursera\footnote{http://www.coursera.org/} as a way to overcome this issue.

For simplicity's sake, we focus on peer review as used in modern scientific communication. The process, as we consider in this paper, can be described as follows: when a manuscript first arrives at the editorial office of an academic journal, it is first examined by the editor, who might reject the manuscript immediately because either it is out-of-scope or because it is of unacceptable quality. Manuscripts that pass this first stage are then sent out to experts with relevant expertise who are usually asked to classify the manuscript as publishable immediately, publishable after some revisions, or not publishable at all. Traditionally, the manuscript's authors do not know the reviewers' identities, but the reviewers may or may not know the identity of the authors.

In other words, peer review can be seen as a decision-making process where the reviewers serve as cognitive inputs that help a decision maker (chair, editor, course instructor, \textit{etc}.) judge the quality of a peer's output. A crucial point in this process is that it greatly depends on the  reviewers' honesty. In the canonical peer-review process, reviewers have no direct incentives for honestly reporting their reviews. Several potential problems have been discussed in different research areas, \textit{e.g.}, bias against female authors, authors from minor institutions, and non-native English writers~\citep{Bornmann:2007, Wenneras:1997, Primack:2008, Newcombe:2009}.

In order to illustrate the application of our method to peer review, we start by modeling the peer-review process as a Bayesian model so as to take the uncertainty regarding the quality of the manuscript into account. We then introduce our scoring method to evaluate reported reviews. We assume that the  scores received by reviewers are somehow coupled with relevant incentives, be they social-psychological, such as praise or visibility, or material rewards through prizes or money. Hence, we naturally assume that reviewers seek to maximize their expected scores and that there are no external incentives. We show that reviewers strictly maximize their expected scores by honestly disclosing their reviews under the additional assumptions that they are Bayesian decision-makers and that they cannot influence the reviews of other reviewers.

Honesty is intrinsically related to accuracy in our peer-review model: as the number of honest reviews increases, the distribution of the reported reviews converges to the probability distribution that represents the quality of the manuscript. We performed peer-review experiments to validate the model and to test the efficiency of the proposed scoring method. Our experimental results corroborate our theoretical model by showing that the act of encouraging honest reporting through the proposed scoring method creates more accurate reviews than the traditional peer-review process, where reviewers have no direct incentives for expressing their true reviews.

In addition to our method for inducing honest reporting, we also propose a method to aggregate opinions that uses information from experts' scores. Our aggregation method is general in a sense that it can be used in any decision-making setting where experts report probability distributions over the outcomes of a discrete random variable. The proposed method works as if the experts were continuously updating their opinions in order to accommodate the expertise of others. Each updated opinion takes the form of a linear opinion pool, where the weight that an expert assigns to a peer's opinion is inversely related to the distance between their opinions. In other words, experts are assumed to prefer opinions that are close to their own opinions, where closeness is defined by an underlying proper scoring rule. We provide conditions under which consensus is achieved under our aggregation method and discuss a behavioral foundation of it. Using data from our peer-review experiments, we find that the consensual review resulting from the proposed aggregation method is consistently more accurate than the canonical average review.

\section{Related Work}
\label{sec:related_work}

In recent years, two prominent methods to induce honest reporting without the assumption of  observable future outcomes were proposed: the \emph{Bayesian truth serum} (BTS) \emph{method} \citep{Prelec:2004} and the \emph{peer-prediction method} \citep{Miller:2005}.

The BTS method works on a single multiple-choice question with a finite number of alternatives. Each expert is requested to endorse the answer mostly likely to be true and to predict the empirical distribution of the endorsed answers. Experts are evaluated by the accuracy of their predictions as well as how surprisingly common their answers are. The surprisingly common criterion exploits the false consensus effect to promote truthfulness, \textit{i.e.}, the general tendency of experts to overestimate the degree of agreement that the others have with them.

The score received by an expert from the BTS method has two major components. The first one, called the information score, evaluates the answer endorsed by the expert according to the log-ratio of its actual-to-predicted endorsement frequencies. The second component, called the prediction score, is a penalty proportional to the relative entropy between the empirical distribution of answers and the expert's prediction of that distribution. Under the BTS scoring method, collective honest reporting is a Bayes-Nash equilibrium.

The BTS method has been used to promote honest reporting in many different domains,  \textit{e.g.}, when sharing rewards amongst a set of experts \citep{Carvalho:2011} and in policy analysis \citep{Weiss:2009}.   However, the BTS method has two major drawbacks. First, it requires the population of experts to be large. Second, besides reporting their opinions, experts must also make predictions about how their peers will report their opinions. While the artificial intelligence community has recently addressed the former issue  \citep{Witkowski:2012, Radanovic:2013}, the latter issue is still an intrinsic requirement for using the BTS method.

The drawbacks of the BTS method are not shared by the peer-prediction method \citep{Miller:2005}. In the peer-prediction method, a number of experts experience a product and rate its quality. A mechanism then collects the ratings and makes payments based on those ratings. The peer-prediction method makes  use of the stochastic correlation between the signals observed by the experts from the product to achieve a Bayes-Nash equilibrium  where every expert reports honestly.

A major problem with the peer-prediction method is that it depends on historical data. For example, when applied to a peer-review setting, after a reviewer $i$ reports his review, say $r_i$, the mechanism then estimates reviewer $i$'s prediction of the review reported by another reviewer $j$, $P(r_j | r_i)$, which is then evaluated and rewarded using a proper scoring rule and reviewer $j$'s actual reported review. The mechanism needs to have a history of previously reported reviews for computing $P(r_j | r_i)$, which is not always a reasonable assumption, \textit{e.g.}, when the evaluation criteria may change from review to review and when the peer-review process is being used for the first time. In other words, the peer-prediction method is prone to cold-start problems.

\cite{Carvalho:2012} addressed this issue by making the extra assumption that experts have uninformative prior knowledge about the distribution of the observed signals. Given this assumption, honest reporting is induced by simply making pairwise comparisons between reported opinions and rewarding agreements. In this paper, we extend the method by \cite{Carvalho:2012} in several ways. First, we show that the assumption of uninformative priors is unnecessary as long as experts have common prior distributions and this fact is common knowledge. Moreover, we provide stronger conditions with respect to the underlying proper scoring rule under which pairwise comparisons induce honest reporting.

Another contribution of our work is a method to aggregate the reported opinions into a single consensual opinion. Over the years, both behavioral and mathematical methods have been proposed to establish consensus \citep{Clemen:1999}. Behavioral methods attempt to generate agreement through interaction and exchange of knowledge. Ideally, the sharing of information leads to a consensus. However, behavioral methods usually provide no conditions under which experts can be expected to reach an agreement. On the other hand, mathematical aggregation methods consist of processes or analytical models that operate on the reported opinions in order to produce a single aggregate opinion. \cite{DeGroot:1974} proposed a model which describes how a group of experts can reach agreement on a consensual opinion by pooling their individual opinions. A drawback of DeGroot's method is that it requires each expert to explicitly assign weights to the opinions of other experts. In this paper, we propose a method to set these weights directly which takes the scores received by the experts into account. We also provide a behavioral interpretation of the proposed aggregation method.

A related method for finding consensus was proposed by  \cite{Carvalho:2013}. Under the assumption that experts prefer probability distributions close to their own distributions, where closeness is measured by the root-mean-square deviation, the authors showed that a consensus is always achieved. Moreover, if risk-neutral experts are rewarded using the quadratic scoring rule, then the assumption that experts prefer probability distributions that are close to their own distributions follows naturally. The approach in this paper is more general because the underlying proper scoring rule can be any bounded proper scoring rule.

From an empirical perspective, we investigate the efficiency of both our scoring method and our method for finding consensus in a peer-review experiment. Formal experiments involving peer review are still relatively scarce. Even though the application of the peer-review process to scientific communication can be traced back almost 300 years, it was not until the early 1990s that research on this matter became more intensive and formalized \citep{vanRooyen:2001}. Scientists in the biomedical domain have been in the forefront of research on the peer-review process due to the fact that dependable quality-controlled information can literally be a matter of life and death in this research field. In particular, the staff of the renowned BMJ, formerly British Medical Journal, have been studying the merits and limitations of peer review over a number of years \citep{Lock:1985, Godlee:2003}. Most of their work has focused on defining and evaluating review quality~\citep{vanRooyen:1999}, and examining the effect of specific interventions on the quality of the resulting reviews~\citep{vanRooyen:2001}.

One mechanism used to prevent bias in the peer-review process is called \emph{double-blind review}, which consists of hiding  both authors and reviewers' identities. Indeed, it has been reported that such a practice reduces bias against female authors \citep{Budden:2008}. However, it can be argued that knowing the authors' identities makes it easier for the reviewers to compare the new manuscript with previously published papers, and it also encourages the reviewers to disclose conflicts of interest. Another argument that undermines the benefits of double-blind reviewing is that the authorship of the manuscript is often obvious to a knowledgeable reader from the context, \textit{e.g.}, self-referencing, research topic, writing style, working paper repositories, seminars, \textit{etc}. \citep{Falagas:2006, Justice:1998, Yankauer:1991}. Furthermore, this mechanism does not prevent against certain types of bias, \textit{e.g.}, when a reviewer rejects new evidence or new knowledge because it contradicts established norms, beliefs or paradigms.

Some work has focused on the calibration aspect of peer review. \cite{Roos:2011} proposed a maximum likelihood method for calibrating reviews by estimating both the bias of each reviewer and the unknown ideal score of the manuscript. Bias is treated as the general rigor of a reviewer across all his reviews. Hence, Roos \textit{et al.}'s method does not attempt to prevent bias by rewarding honest reporting. Instead, it adjusts reviews \emph{a posteriori} so that they can be globally comparable.

Instead of calibrating reviews \emph{a posteriori}, \cite{robinson:2001} suggested  to ``calibrate" reviewers \textit{a priori}. Reviewers are first asked to review short texts that have gold-standard reviews, \textit{i.e.}, reviews of high quality provided by experts with relevant expertise. Thereafter, they receive calibration scores, which are later used as weighting factors to determine how well their future reviews will be considered. This approach, however, does not guarantee that reviewers will report honestly after the calibration phase, when gold-standard reviews are no longer available.

To the best of our knowledge, our peer-review experiments are the first to investigate the use of incentives for honest reporting in a peer-review task. When objective verification is not possible, as in the peer-review process, economic measures may be used to encourage experts to honestly disclose their opinions. The proposed scoring method does so by making pairwise comparisons between reported reviews and rewarding agreements.

Rewarding experts based on pairwise comparisons has been empirically proven to be an effective incentive technique in other domains. \cite{Shaw:2011} measured the effectiveness of a collection of social and financial incentive schemes for motivating experts to conduct a qualitative content analysis task. The authors found that treatment conditions that provided financial incentives and asked experts to prospectively think about the responses of their peers  produced more accurate responses. \cite{Huang:2013} showed that informing the experts that their rewards will be based on how similar their responses are to other experts' responses produces more accurate responses than telling the experts that their rewards will be based on how similar their responses are to gold-standard responses. Our work adds to the existing body of literature by theoretically and empirically showing that pairwise comparisons make the peer-review process more accurate.

\section{The Basic Model}
\label{sec:model}

In our proposed peer-review process, a \emph{manuscript} is reviewed by a set of \emph{reviewers} $N = \{1, \dots, n\}$, with $n \geq 2$. The quality of the manuscript is represented by a multinomial distribution\footnote{We use the term \emph{multinomial distribution}  to refer to the generalization of the Bernoulli distribution for discrete random variables with any constant number of outcomes. The parameter of this distribution is a probability vector that specifies the probability of each possible outcome.} $\Omega$ with \emph{unknown} parameter $\boldsymbol\omega = (\omega_{0}, \dots, \omega_{v})$, where $v \in \mathbb{N}^+$ represents the best \emph{evaluation score} that the manuscript can receive and $\omega_{k}$ is the probability assigned to the evaluation score being equal to $k$.

Each reviewer is modeled  as possessing a privately observed draw (signal) from $\Omega$. Hence, our model captures the uncertainty of the reviewers regarding the quality of the manuscript. We extend the model to multiple observed signals in Section 5. We denote the \emph{honest review} of each reviewer $i \in N$ by $t_i \sim \Omega$, where $t_i \in \{0, \dots, v\}$. Honest reviews are independent and identically distributed, \textit{i.e.}, $P(t_i | t_j) = P(t_i)$. We say that reviewer $i$ is reporting honestly when his \emph{reported review} $r_i$ is equal to his honest review, \textit{i.e.}, $r_i = t_i$.

Reviews are elicited and aggregated by a trusted entity referred  to as the \emph{center}\footnote{We refer to a single reviewer as ``he" and to the center as ``she".}, which is also responsible for rewarding the reviewers. Let $s_i$ be reviewer $i$'s  \emph{review score} after he reports $r_i$. We discuss how $s_i$ is determined in Section 4. Review scores are somehow coupled with relevant incentives, be they social-psychological, such as praise or visibility, or material rewards through prizes or money. We make four major assumptions in our model:

\begin{enumerate}

 \item \emph{Autonomy}: Reviewers cannot influence other reviewers' reviews, \textit{i.e.}, they do not know each other's identity and they are not allowed to communicate to each other during the reviewing process.

 \item \emph{Risk Neutrality}: Reviewers behave so as to maximize their expected review scores.
 
 \item \emph{Dirichlet Priors}: There exists a common prior distribution over $\boldsymbol\omega$, \textit{i.e.}, $P\left(\boldsymbol\omega\right)$. We assume that this prior is a Dirichlet distribution and this is common knowledge.

 \item \emph{Rationality}: After observing $t_i$, every reviewer $i \in N$ updates his belief by applying Bayes' rule to the common prior, \textit{i.e.}, $P\left( \mathbf{\boldsymbol\omega} | {t}_i \right)$.
 
\end{enumerate}

The first assumption describes how peer review is traditionally done in practice. The second assumption means that reviewers are self-interested and no external incentives exist for each reviewer. The third assumption means that reviewers have common prior knowledge about the quality of the manuscript, a natural assumption in the peer-review process. We discuss the formal meaning  of such an assumption in the following subsection. The fourth assumption implies that the posterior distributions are consistent with Bayesian updating, \textit{i.e.}:

\begin{equation*}
P\left(\boldsymbol\omega | t_{i}\right) = \frac{P(t_{i} | \boldsymbol\omega)P(\boldsymbol\omega)}{P(t_{i})} 
\end{equation*}

The last three assumptions imply that reviewers are \emph{Bayesian decision-makers}. We note that different modeling choices could have been used, \textit{e.g.}, models based on games of incomplete information. Unlike our model, an incomplete-information game is often used when experts do not know each other's beliefs. To find strategic equilibria in such incomplete-information models, one would need information about experts' beliefs about each other's private information. A Bayesian structure could be used to model each expert's beliefs about the others, and it would permit the calculation of experts' expected scores, which are maximized at equilibrium. However, the natural autonomy assumption makes such a Bayesian structure unrealistic.

\subsection{Dirichlet Distributions}
\label{subsec:dirichlet}

An important assumption in our model is that reviewers have  \emph{Dirichlet priors} over distributions of evaluation scores. The \emph{Dirichlet distribution} can be seen as a continuous distribution over parameter vectors of a multinomial distribution. Since $\boldsymbol\omega$ is the unknown parameter of the multinomial distribution that describes the quality of the manuscript, then it is natural to consider a Dirichlet distribution as a prior for $\boldsymbol\omega$. Given a vector of positive integers, $\boldsymbol\alpha = (\alpha_0, \dots , \alpha_v)$, that determines the shape of the Dirichlet distribution, the probability density function of the Dirichlet distribution over $\boldsymbol\omega$ is:

\begin{equation}
\label{eq:dirichlet_dist}
P(\boldsymbol\omega  |  \boldsymbol\alpha) = 
\frac{1}{\beta(\boldsymbol\alpha)}\prod_{k=0}^v \omega_{k}^{\alpha_k - 1}
\end{equation}

\noindent where:

\begin{equation*}
\beta(\boldsymbol\alpha) = \frac{\prod_{k=0}^v (\alpha_k -1)!}
                 {\left(\sum_{k=0}^v \alpha_k - 1\right)!}
\end{equation*}

Figure~\ref{fig:Dirichlet} shows the above probability density when $v=2$ for some parameter vectors $\boldsymbol\alpha$. For the Dirichlet distribution in (\ref{eq:dirichlet_dist}), the expected value of $\omega_{j}$ is $\mathbb{E}[\omega_{j} |  \boldsymbol\alpha] = \alpha_i/\sum_{k=0}^v \alpha_k$. The probability vector $\mathbb{E}[\boldsymbol\omega |  \boldsymbol\alpha]=(\mathbb{E}[\omega_{0} |  \boldsymbol\alpha], \dots, \mathbb{E}[\omega_{v} | \boldsymbol\alpha])$ is called the \emph{expected distribution} regarding $\boldsymbol\omega$.

An interesting property of the Dirichlet distribution is that it is the \emph{conjugate prior} of the multinomial distribution \citep{Bernardo:1994}, \textit{i.e.},  the posterior distribution $P(\boldsymbol\omega | \boldsymbol\alpha, t_i)$ is itself a Dirichlet distribution. This relationship is often used in Bayesian statistics to estimate  hidden parameters of multinomial distributions. To illustrate this point, suppose that reviewer $i$ observes the signal $t_i = x$, for $x \in \{0, \dots, v\}$. After applying Bayes' rule, reviewer $i$'s posterior distribution is $P(\boldsymbol\omega | \boldsymbol\alpha, t_i = x) = P(\boldsymbol\omega | (\alpha_0, \alpha_1, \dots, \alpha_x + 1, \dots \alpha_v))$. Consequently, the new expected distribution is:

\begin{equation}
\label{eq:pos_pred_dist}
\mathbb{E}[\boldsymbol\omega |  \boldsymbol\alpha, t_i=x] = \left(
\frac{\alpha_0}{1 + \sum_{k=0}^v \alpha_k},
\frac{\alpha_1}{1 + \sum_{k=0}^v \alpha_k}, 
\dots,
\frac{\alpha_x + 1}{1 + \sum_{k=0}^v \alpha_k}, 
\dots,
\frac{\alpha_v}{1 + \sum_{k=0}^v \alpha_k}\right) 
\end{equation}

We call the probability vector in (\ref{eq:pos_pred_dist}) reviewer $i$'s \emph{posterior predictive distribution} regarding $\boldsymbol{\omega}$ because it provides the distribution of future outcomes given the observed data $t_i$. With this perspective, we regard the values $\alpha_0, \dots , \alpha_v$ as ``pseudo-counts" from ``pseudo-data", where each $\alpha_k$ can be interpreted as the number of times that the $\omega_{k}$-probability event has been observed before.

\begin{figure}
\begin{center}
\includegraphics[width=\textwidth]{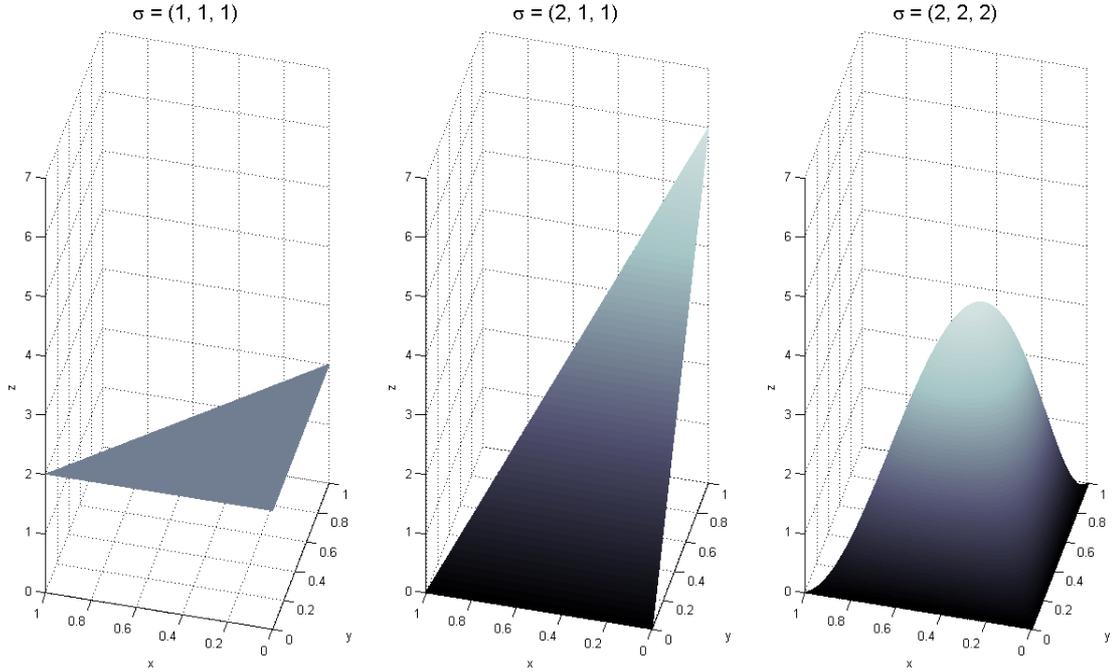}
\caption{Probability densities of Dirichlet distributions when $v=2$ for different parameter vectors. Left: $\boldsymbol\alpha = (1, 1, 1)$. Center: $\boldsymbol\alpha = (2, 1, 1)$. Right: $\boldsymbol\alpha = (2, 2, 2)$. }
\label{fig:Dirichlet}
\end{center}
\end{figure}

Throughout this paper, we assume that reviewers have common prior Dirichlet distributions and this fact is common knowledge, \textit{i.e.}, the value of $\boldsymbol{\alpha}$ is initially the same for all reviewers. A practical interpretation of this assumption is that reviewers have common prior knowledge about the quality of the manuscript, \textit{i.e.}, reviewers have a common expectation regarding the quality of arriving manuscripts.

By using Dirichlet distributions as priors, belief updating can be expressed as an updating of the parameters of the prior distribution\footnote{We note that other priors could have been used. However, the inference process would not necessarily be analytically tractable. In general, tractability can be obtained through conjugate distributions. Hence, another modeling choice is to consider that  evaluation scores follow a normal distribution with unknown parameters. Assuming exchangeability, we can then use either the normal-gamma distribution or the normal-scaled inverse gamma distribution as the conjugate prior~\citep{Bernardo:1994}. The major drawback with this approach is that continuous evaluation scores might bring extra complexity to the reviewing process.}. Furthermore, the assumption of common knowledge allows the center to estimate reviewers' posterior distributions based solely on their reported reviews, a point which is explored by our proposed scoring method. Due to its attractive theoretical properties, the Dirichlet distribution has been used to model uncertainty in a variety of different scenarios, \textit{e.g}, when experts are sharing a reward based on peer evaluations \citep{Carvalho:2012} and when experts are grouped based on their  individual differences \citep{Navarro:2006}.

\section{Scoring Method}

In this section, we propose a scoring method to induce honest reporting of reviews. The proposed method is built on proper scoring rules \citep{Winkler:1968}.

\subsection{Proper Scoring Rules}
\label{sec:scoring_rules}

Consider an uncertain quantity with possible outcomes $o_0, \dots, o_v$, and a probability vector $\mathbf{z} = (z_0, \dots, z_v)$, where $z_k$ is the probability value associated with the occurrence of outcome $o_k$. A \emph{scoring rule} $R(\mathbf{z}, e)$ is a function that provides a score for the assessment $\mathbf{z}$ upon observing the outcome $o_e$, for $e \in \{0, \dots, v\}$. A scoring rule is called  \emph{strictly proper} when an expert receives his maximum expected score if and only if his stated assessment $\mathbf{z}$ corresponds to his true assessment  $\mathbf{q} = (q_0,\dots, q_v)$~\citep{Winkler:1968}. The \emph{expected score} of $\mathbf{z}$ at $\mathbf{q}$ for a real-valued scoring rule $R(\mathbf{z}, e)$ is:

\begin{equation*}
\mathbb{E}_{\mathbf{q}} \left[ R(\mathbf{z}, e) \right] = \sum_{e=0}^v q_e\, R(\mathbf{z},e)
\end{equation*}

Proper scoring rules have been used as a tool to promote honest reporting in a variety of domains, \textit{e.g.}, when sharing rewards amongst a set of experts based on peer evaluations \citep{Carvalho:2010,Carvalho:2012}, to incentivize experts to accurately estimate their own efforts to accomplish a task \citep{Bacon:2012}, in prediction markets \citep{Hanson:2003}, in weather forecasting \citep{Gneiting:2007}, \textit{etc}. Some of the best known strictly proper scoring rules, together with their scoring ranges, are:

\begin{align}
\mbox{logarithmic: }& R(\textbf{z},e) = \log z_e \hspace{0.8in} (-\infty,0] \nonumber\\
\mbox{quadratic: }  & R(\textbf{z},e) = 2z_e - \sum_{k=0}^v z_{k}^2 \hspace{0.46in}  [-1,1]\label{eq:quad_scor_rul}\\
\mbox{spherical: }  & R(\textbf{z},e) = \frac{z_e}{\sqrt{\sum_{k=0}^v z_{k}^2}} \hspace{0.63in} [0,1] \nonumber
\end{align}

All the above scoring rules are \emph{symmetric}, \textit{i.e.}, $R((z_0, \dots, z_v), e) =  R((z_{\pi_0}, \dots, z_{\pi_v}), \pi_e)$, for all probability vectors $\mathbf{z} = (z_0, \dots, z_v)$, for all permutations $\pi$ on $v+1$ elements, and for all outcomes indexed by $e \in \{0, \dots, v\}$.  We say that a scoring rule is \emph{bounded} if $R(\mathbf{z}, e) \in \mathbb{R}$, for all probability vectors $\mathbf{z}$ and $e \in \{0, \dots, v\}$. For example, the logarithmic scoring rule is not bounded  because it might return $-\infty$ whenever the probability vector $\mathbf{z}$ contains an element equal to zero, whereas both the quadratic and spherical scoring rules are always bounded. A well-known property of strictly proper scoring rules is that they are still strictly proper under positive affine transformations \citep{Gneiting:2007}, \textit{i.e.}, $\mbox{argmax}_{\mathbf{z}} \mathbb{E}_{\mathbf{q}} \left[ \gamma R(\mathbf{z}, e) + \lambda\right] = \mbox{argmax}_{\mathbf{z}} \mathbb{E}_{\mathbf{q}} \left[ R(\mathbf{z}, e)\right] = \mathbf{q}$, for a strictly proper scoring rule $R$, $\gamma > 0$, and $\lambda \in \mathbb{R}$.

\begin{proposition} 
If $R(\mathbf{z}, e)$ is a strictly proper scoring rule, then a positive affine transformation of $R$, \textit{i.e.}, $\gamma R(\mathbf{z}, e) + \lambda,$ for $\gamma > 0$ and $\lambda \in \mathbb{R}$, is also strictly proper.
\end{proposition}

\subsection{Review Scores}
\label{sec:review_scores}

If we knew \textit{a priori} reviewers' honest reviews, we could then compare the honest reviews to the reported reviews and reward agreement. However, due to the subjective nature of the peer-review process, we are facing a situation where this objective truth is practically unknowable. Our solution is to induce honest reporting through pairwise comparisons of reported reviews. The first step towards computing each reviewer $i$'s review score is to estimate  his posterior predictive distribution $\mathbb{E}[\boldsymbol\omega |  \boldsymbol\alpha, t_i]$ shown in (\ref{eq:pos_pred_dist}) based on his reported review $r_{i}$. Let $\mathbb{E}[\boldsymbol\omega |  \boldsymbol\alpha, r_i] = (\mathbb{E}[\omega_0 |  \boldsymbol\alpha, r_i], \dots,\mathbb{E}[\omega_v |  \boldsymbol\alpha, r_i])$  be such an estimation, where:

\begin{equation}
\label{eq:est_post_pred_dist}
\mathbb{E}[\omega_k |  \boldsymbol\alpha, r_i] = \left\{ \begin{array}{ll}
\frac{\alpha_k +1}{1 + \sum_{x=0}^v \alpha_x} & \textrm{if $r_{i} = k$},\\
\frac{\alpha_k}{1+ \sum_{x=0}^v \alpha_x} & \textrm{otherwise}.\\
\end{array} \right.
\end{equation}

Recall that the elements of reviewer $i$'s true posterior predictive distribution are defined as:

\begin{equation*}
\mathbb{E}\left[\omega_{k} |  \boldsymbol{\alpha}, t_{i}\right] =  \left\{ \begin{array}{ll}
\frac{\alpha_k +1}{1 + \sum_{x=0}^v \alpha_x} & \textrm{if $t_{i} = k$},\\
\frac{\alpha_k}{1+ \sum_{x=0}^v \alpha_x} & \textrm{otherwise}.\\
\end{array} \right.
\end{equation*}

Clearly, $\mathbb{E}\left[\omega_{k} |  \boldsymbol{\alpha}, r_{i}\right] =  \mathbb{E}\left[\omega_{k} |  \boldsymbol{\alpha}, t_{i}\right] $  if and only if reviewer $i$ is reporting honestly, \textit{i.e.}, when he reports $r_i = t_i$. The review score of reviewer $i$ is determined as follows:

\begin{equation}
\label{eq:truth_telling_score}
s_i = \sum_{j \neq i} ( \gamma R\left(\mathbb{E}\left[\boldsymbol\omega |  \boldsymbol{\alpha}, r_{i}\right], r_{j} \right) + \lambda)
\end{equation}

\noindent where $\gamma$ and $\lambda$ are constants, for $\gamma > 0$ and $\lambda \in \mathbb{R}$, and $R$ is a strictly proper scoring rule. Scoring rules require an observable outcome, or a ``reality", in order to score an assessment. Intuitively, we consider each review reported by every reviewer other than reviewer $i$ as an observed outcome, \textit{i.e.}, the evaluation score deserved by the manuscript, and then we score reviewer $i$'s \emph{estimated posterior predictive distribution} in (\ref{eq:est_post_pred_dist}) as an assessment of that value.

\begin{proposition}
Each reviewer $i \in N$ strictly maximizes his expected review score if and only if $r_{i} = t_{i}$.
\end{proposition}

\begin{proof}
Let $\mathbf{\Theta}_i = \mathbb{E}\left[\boldsymbol\omega |  \boldsymbol{\alpha}, t_{i}\right]$ and $\mathbf{\Phi}_i = \mathbb{E}\left[\boldsymbol\omega |  \boldsymbol{\alpha}, r_{i}\right]$. By the autonomy assumption, reviewers cannot affect their peers' reviews. Hence, we can restrict ourselves to show that each reviewer $i \in N$ strictly maximizes $\mathbb{E}_{\mathbf{\Theta}_i} \left[\gamma R\left(\mathbf{\Phi}_i, r_j\right)+\lambda\right]$, for $j\neq i$, if and only if $r_i = t_i$.

\textbf{(If part)} Since $R$ is a strictly proper scoring rule, from Proposition 1 we have that:

\begin{equation*}
\argmax_{\mathbf{\Phi}_i} \mathbb{E}_{\mathbf{\Theta}_i} \left[\gamma R\left(\mathbf{\Phi}_i, r_j \right) + \lambda\right] = \mathbf{\Theta}_i
\end{equation*}

If $r_{i} = t_{i}$, then by construction $\mathbf{\Phi}_i = \mathbf{\Theta}_i$, \textit{i.e.}, the estimated posterior predictive distribution in (\ref{eq:est_post_pred_dist}) is equal to the true posterior predictive distribution in (\ref{eq:pos_pred_dist}). Consequently, honest reporting strictly maximizes reviewers' expected review scores.

\textbf{(Only-if part)}. Using a similar argument, given that $R$ is a strictly proper scoring rule, from Proposition 1 we have that:

\begin{equation*}
\argmax_{\mathbf{\Phi}_i} \mathbb{E}_{\mathbf{\Theta}_i} \left[\gamma R\left(\mathbf{\Phi}_i, r_j \right) + \lambda\right] = \mathbf{\Theta}_i
\end{equation*}  

By construction, $\mathbf{\Phi}_i = \mathbf{\Theta}_i $ if and only if $r_i = t_i$ (see equations (\ref{eq:pos_pred_dist}) and (\ref{eq:est_post_pred_dist})). Thus, reviewers' expected review scores are strictly maximized only when reviewers are honest. 
\end{proof}

Another way to interpret the above result is to imagine that each reviewer is betting on the review deserved by the manuscript. Since the most relevant information available to him is the observed signal, then the strategy that maximizes his expected review score is to bet on that signal, \textit{i.e.}, to bet on his honest review. When this happens, the true posterior predictive distribution in (\ref{eq:pos_pred_dist}) is equal to the estimated posterior predictive distribution in (\ref{eq:est_post_pred_dist}) and, consequently, the expected score resulting from a strictly proper scoring rule is strictly maximized when the expectation is taken with respect to the true posterior predictive distribution.

It is important to observe that by incentivizing honest reporting, the scoring function in (\ref{eq:truth_telling_score}) also incentivizes accuracy since honest reviews are draws from the distribution that represents the true quality of the manuscript. In other words, the center is indirectly observing these draws when reviewers report honestly. Consequently, due to the law of large numbers, the distribution of the reported reviews converges to the distribution that represents the true quality of the manuscript as the number of honestly reported reviews increases. Our experimental results in Section 7 show that there indeed exists a strong correlation between honesty and accuracy.

Different interpretations of the scoring method in (\ref{eq:truth_telling_score}) arises depending on the underlying strictly proper scoring rule and the hyperparameter $\boldsymbol{\alpha}$. In the following subsections, we discuss two different interpretations: 1) when $R$ is a symmetric and bounded strictly proper scoring rule and reviewers' prior distributions are non-informative; and 2) when $R$ is a strictly proper scoring rule sensitive to distance.

\subsection{Rewarding Agreement}
\label{sec:agreement}

Assume that reviewers' prior distributions are non-informative, \textit{i.e.}, all the elements making up the hyperparameter $\boldsymbol{\alpha}$ have the same value. This happens when reviewers have no relevant prior knowledge about the quality of the manuscript. Consequently, the elements of reviewers' true and estimated posterior predictive distributions can take on only two possible values (see equations (\ref{eq:pos_pred_dist}) and (\ref{eq:est_post_pred_dist}) for $\alpha_0 = \alpha_1 = \dots = \alpha_v$ ).

Moreover, if $R$ is a symmetric scoring rule, then the term $ R\left(\mathbb{E}\left[\boldsymbol\omega |  \boldsymbol{\alpha}, r_{i}\right], r_{j} \right)$  in (\ref{eq:truth_telling_score}) can take on only two possible values because a permutation of elements with similar values does not change the score of a symmetric scoring rule. When $R$ is also strictly proper, it means that $ R\left(\mathbb{E}\left[\boldsymbol\omega |  \boldsymbol{\alpha}, r_{i}\right], r_{j} \right) = \delta_{max}$, when $r_i = r_j$, and $ R\left(\mathbb{E}\left[\boldsymbol\omega |  \boldsymbol{\alpha}, r_{i}\right], r_{j} \right) = \delta_{min}$, when $r_i \neq r_j$, where $\delta_{max} > \delta_{min}$. Consequently, each term of the summation in (\ref{eq:truth_telling_score}) can be written as:

\begin{equation*}
\gamma R\left(\mathbb{E}\left[\boldsymbol\omega |  \boldsymbol{\alpha}, r_{i}\right], r_{j} \right) + \lambda = \left\{ \begin{array}{ll}
\gamma \delta_{max} + \lambda   & \textrm{if $r_{i} = r_{j}$},\\
\gamma \delta_{min} + \lambda & \textrm{otherwise}.\\
\end{array} \right.
\end{equation*}

When $R$ is also bounded, we can then set $\gamma = \frac{1}{\delta_{max} - \delta_{min}}$ and $\lambda = \frac{-\delta_{min}}{\delta_{max} - \delta_{min}}$, and the above values become, respectively, $1$ and $0$. Hence, the resulting review scores do not depend on parameters of the model. Moreover, we obtain an  intuitive interpretation of the scoring method in (\ref{eq:truth_telling_score}): whenever two reported reviews are equal to each other, the underlying reviewers are rewarded by one payoff unit. Thus, in practice, our scoring method works by simply comparing reported reviews and rewarding agreements whenever $R$ is a symmetric and bounded strictly proper scoring rule and reviewers have no informative prior knowledge about the quality of the manuscript.

Another interesting point is that the center can reward different agreements in different ways, \textit{i.e.}, reviewers are not necessarily equally valued. For example, if the center knows \emph{a priori} that a particular reviewer $j$ is reliable (respectively, unreliable), then she can increase (respectively, decrease) the reward of reviewers whose reviews are in agreement with reviewer $j$'s reported review. Formally, this means that for different reviewers $i$ and $j$, the center can use different values for $\gamma$ and $\lambda$  in  (\ref{eq:truth_telling_score}). Proposition 2 is not affected by this as long as $\gamma > 0$, $\lambda \in \mathbb{R}$, and their values are independent of the reported reviews. Hence, by having a few reliable reviewers, this approach might help to eliminate the hypothetical scenario where a set of reviewers learn over time to report similar reviews. A similar idea was proposed by \cite{Jurca:2009} to prevent collusions in reputation systems.

\subsection{Strictly Proper Scoring Rules Sensitive to Distance}
\label{subsec:spsr_distance}

Pairwise comparisons, as defined in the previous subsection, might work well for small values of $v$, the best evaluation score that the manuscript can receive, but it can be too restrictive and, to some degree, unfair when the best evaluation score is high. For example, when $v = 10$ and the review used as the observed outcome is also equal to $10$, a reported review equal to $9$ seems to be more accurate than a reported review equal to $1$. One effective way to deal with these issues is by using strictly proper scoring rules in (\ref{eq:truth_telling_score}) that are \emph{sensitive to distance}.

Using the notation of Section 4.1, recall that $\mathbf{z} = (z_0,\dots, z_v)$ is some reported probability distribution. Given that the outcomes are ordered, we denote the cumulative probabilities by capital letter: ${Z}_k = \sum_{j\leq k} z_j$.  We first define the notion of distance between two probability vectors as proposed by  \cite{Holstein:1970}. We say that a probability vector $\mathbf{z}^\prime$ is more distant from the $j$th outcome than a probability vector $\mathbf{z} \neq \mathbf{z}^\prime$ if:

\begin{align}
{Z}_k^\prime \geq {Z}_k, &\mbox{  for } k = 0, \dots, j-1\nonumber\\
{Z}_k^\prime \leq {Z}_k, &\mbox{  for } k = j, \dots, v \nonumber
\end{align}

Intuitively, the above definition means that $\mathbf{z}$ can be obtained from $\mathbf{z}^\prime$ by successively moving probability mass towards the $j$th outcome from other outcomes \citep{Holstein:1970}. A scoring rule $R$  is said to be \emph{sensitive to distance} if  $R(\mathbf{z}, j) > R(\mathbf{z}^\prime, j)$ whenever $\mathbf{z}^\prime$ is more distant from $\mathbf{z}$ for all $j$. \cite{Epstein:1969} introduced the \emph{ranked probability score} (RPS), a strictly proper scoring rule that is sensitive to distance. Using the formulation of Epstein's result proposed by \cite{Murphy:1971}, we have for a probability vector $\mathbf{z}$ and an observed outcome $j \in \{0, \dots, v\}$:

\begin{equation}
\label{eq:RPS}
RPS(\mathbf{z}, j) = - \sum_{k=0}^{j-1} {Z}_k^2 - \sum_{k=j}^{v} (1-{Z}_k)^2
\end{equation}

Figure~\ref{fig:sco_rul_dist} illustrates the scores returned by (\ref{eq:RPS}) for different reported reviews and values for $j$ when reviewers' prior distributions are non-informative. When using RPS as the strictly proper scoring rule in (\ref{eq:truth_telling_score}), reviewers are rewarded based on how close their reported reviews are to the  reviews taken as observed outcomes. For example, when the review used as the observed outcome is equal to $0$ (see the dotted line with squares in Figure~\ref{fig:sco_rul_dist}), the returned score monotonically decreases as the reported review increases. Since RPS is strictly proper, Proposition 2 is still valid for any hyperparameter $\boldsymbol\alpha$, \textit{i.e.}, each reviewer strictly maximizes his expected review score by reporting honestly. The scoring range of RPS is $[-v, 0]$. Hence, review scores are always non-negative when using $\gamma = 1$ and $\lambda = v$ in (\ref{eq:truth_telling_score}).

\begin{figure}
\centering
\includegraphics[scale=0.4]{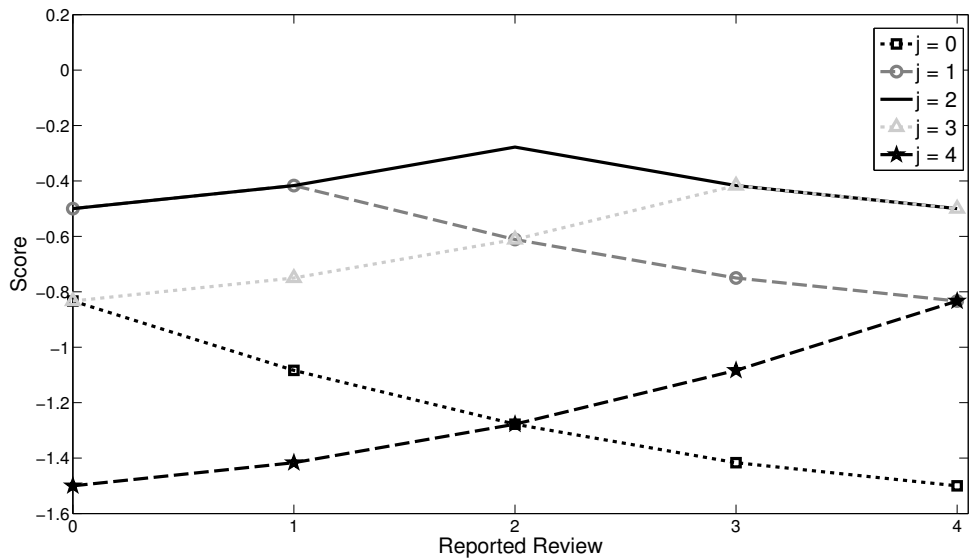}
\caption{Scores returned by $R\left(\mathbb{E}\left[\boldsymbol\omega |  \boldsymbol{\alpha}, r_{i}\right], j \right)$ for different reported reviews when $v = 4$, $R$ is the RPS rule, and $\boldsymbol\alpha=(1,1,1,1,1)$. Each line represents a different value for $j$ (observed outcome).}
\label{fig:sco_rul_dist}
\end{figure}

\subsection{Numerical Example}
\label{subsec:num_example}

Consider four reviewers ($n = 4$) and the best evaluation score being equal to four ($v = 4$). Suppose that reviewers have non-informative Dirichlet priors with $\boldsymbol\alpha = (1,1,1,1,1)$, and that reviewers 1, 2, 3, and 4 report, respectively, ${r}_{1} = 0$, ${r}_{2} = 0$, $r_{3} = 1$, and $r_{4} = 4$. From (\ref{eq:est_post_pred_dist}), the resulting estimated posterior predictive distributions are, respectively, $\mathbb{E}[\boldsymbol\omega |  \boldsymbol\alpha, r_1=0] = \left(\frac{2}{6}, \frac{1}{6}, \frac{1}{6}, \frac{1}{6}, \frac{1}{6} \right)$, $\mathbb{E}[\boldsymbol\omega|  \boldsymbol\alpha, r_2=0] = \left(\frac{2}{6}, \frac{1}{6}, \frac{1}{6}, \frac{1}{6}, \frac{1}{6} \right)$, $\mathbb{E}[\boldsymbol\omega |  \boldsymbol\alpha, r_3=1] = \left(\frac{1}{6}, \frac{2}{6}, \frac{1}{6}, \frac{1}{6}, \frac{1}{6} \right)$, and $\mathbb{E}[\boldsymbol\omega |  \boldsymbol\alpha, r_4=4] = \left(\frac{1}{6}, \frac{1}{6}, \frac{1}{6}, \frac{1}{6}, \frac{2}{6} \right)$. In what follows, we illustrate the scores returned by (\ref{eq:truth_telling_score}) when using a symmetric and bounded strictly proper scoring rule and when using RPS.

\subsubsection{Rewarding Agreements}

Assume that $R$ in (\ref{eq:truth_telling_score}) is the quadratic scoring rule shown in (\ref{eq:quad_scor_rul}), which in turn is symmetric, bounded, and strictly proper. Consequently, as discussed in Section 4.3, the term $\gamma R\left(\mathbb{E}[\boldsymbol\omega |  \boldsymbol\alpha, r_i], r_j \right) + \lambda$ in (\ref{eq:truth_telling_score}) can take on only two values:

\begin{equation*}
\begin{array}{ll}
\gamma \left( \frac{4}{v+2} - \left(\frac{2}{v+2}\right)^2 - \sum_{e=0}^{v-1}\left(\frac{1}{v+2}\right)^2\right) + \lambda   & \textrm{if $r_{i} = r_{j}$},\\
\gamma \left( \frac{2}{v+2} - \left(\frac{2}{v+2}\right)^2 - \sum_{e=0}^{v-1}\left(\frac{1}{v+2}\right)^2\right) + \lambda  & \textrm{otherwise}.\\
\end{array} 
\end{equation*}

Hence, by setting $\gamma = \frac{1}{\delta_{max} - \delta_{min}} = \frac{v+2}{2}$ and $\lambda = \frac{-\delta_{min}}{\delta_{max} - \delta_{min}} = \frac{-v}{2v+4}$, the above values are equal to, respectively, $1$ and $0$. Using the scoring method in (\ref{eq:truth_telling_score}), we obtain the following review scores: $s_{1} = s_{2} = 1$  and $s_{3} = s_{4} = 0$. That is, the review scores received by reviewers $1$ and $2$ are similar due to the fact that $r_1 = r_2$. Reviewer $3$ and $4$'s review scores are equal to $0$ because there is no match between their reported reviews and others' reported reviews.

\subsubsection{Taking Distance into Account}

Now, assume that $R$ in (\ref{eq:truth_telling_score}) is the RPS rule shown in (\ref{eq:RPS}). In order to ensure non-negative review scores, let $\gamma = 1$ and $\lambda = v = 4$ . Using the scoring method in (\ref{eq:truth_telling_score}), we obtain the following review scores: $s_1 = s_2 = (-0.8333\gamma + \lambda) + (-0.5\gamma + \lambda) + (-1.5\gamma + \lambda) = 9.1667, s_{3} = 2\cdot(-1.0833\gamma + \lambda) + (-1.4167\gamma + \lambda) = 8.4167$, and $s_{4} = 2\cdot(-1.5\gamma + \lambda) + (-0.8333\gamma + \lambda) = 8.1667$. The review score of reviewer $4$ is the lowest because his reported review is the most different review, \textit{i.e.}, it has the largest distance between it and all of the other reviews.

\section{Multiple Criteria}

In our basic model, reviewers observe only one signal from the distribution that represents the quality of the manuscript. However, manuscripts are often evaluated under multiple criteria, \textit{e.g.}, relevance, clarity, originality, \textit{etc}., meaning that in practice reviewers might observe multiple signals and report multiple evaluation scores. Under the assumption that these signals are independent, each reported evaluation score can be scored individually using the same scoring method proposed in the previous section. Clearly, Proposition 2 is still valid, \textit{i.e.}, honest reporting still strictly maximizes reviewers' expected review scores.

The lack of relationship between different criteria is not always a reasonable assumption. A modeling choice that takes the relationship between observed signals into account, which is also consistent with our basic model, is to assume that the quality of the manuscript is still represented by a  multinomial distribution, but now reviewers may observe several signals from that distribution. Formally, let $\rho \in \mathbb{N}^+$ be the number of draws from the distribution that represents the quality of the manuscript, where each signal represents an evaluation score related to a criterion. Instead of a single number, each reviewer $i$'s private information is now a vector:  $\mathbf{t}_i = \left(t_{i, 1}, \dots, t_{i, \rho}\right)$, where $t_{i,k} \in \{0, \dots, v\}$, for $k \in \{1,\dots, \rho\}$. The basic assumptions (autonomy, risk neutrality, Dirichlet priors, and  rationality) are still the same. For ease of exposition, we denote reviewer $i$'s true posterior predictive distribution in this section by $\mathbf{\Theta}_i^{(\rho)} = \mathbb{E}\left[\boldsymbol\omega |  \boldsymbol{\alpha}, \mathbf{t}_{i}\right]$. Under this new model, each reviewer $i$'s posterior predictive distribution is now defined as:

\begin{equation}
\label{eq:pos_pred_dist_mult_obs}
\mathbf{\Theta}_i^{(\rho)}  = \left(\frac{\alpha_0 + \sum_{k = 1}^\rho H(0, t_{i ,k})}
                                      {\rho+\sum_{x=0}^v \alpha_x},
                          \frac{\alpha_1 + \sum_{k = 1}^\rho H(1, t_{i ,k})}
                                      {\rho+\sum_{x=0}^v \alpha_x}, \dots,
                          \frac{\alpha_v + \sum_{k = 1}^\rho H(v, t_{i, k})}                                                            
                                      {\rho+\sum_{x=0}^v \alpha_x}\right)
\end{equation}

\noindent where $H(x, y)$ is an indicator function:

\begin{equation*}
H(x, y) = \left\{ \begin{array}{ll}
1 & \textrm{if $x = y$,}\\
0 & \textrm{otherwise}.\\
\end{array} \right.
\end{equation*}

Assuming that each reported review is a vector of $\rho$ evaluations scores, \textit{i.e.}, $\mathbf{r}_i = \left(r_{i, 1}, \dots, r_{i, \rho}\right)$, where $r_{i,k} \in \{0, \dots, v\}$, for $i \in N$ and $k \in \{1,\dots, \rho\}$, the center estimates each reviewer $i$'s posterior predictive distribution by applying Bayes' rule to the common prior. The resulting estimated posterior predictive distribution $\mathbb{E}\left[\boldsymbol\omega |  \boldsymbol{\alpha}, \mathbf{r}_{i}\right]$, referred to as $\mathbf{\Phi}_{i}^{(\rho)}$ for ease of exposition, is:

\begin{equation}
\label{eq:est_pos_pred_dist_mult_obs}
\mathbf{\Phi}_{i}^{(\rho)} = \left(\frac{\alpha_0 + \sum_{k = 1}^\rho H(0, r_{i ,k})}
                                      {\rho+\sum_{x=0}^v \alpha_x},
                          \frac{\alpha_1 + \sum_{k = 1}^\rho H(1, r_{i ,k})}
                                      {\rho+\sum_{x=0}^v \alpha_x}, \dots,
                          \frac{\alpha_v + \sum_{k = 1}^\rho H(v, r_{i, k})}                                                            
                                      {\rho+\sum_{x=0}^v \alpha_x}\right)
\end{equation}

Thereafter, the center rewards $\mathbf{\Phi}_{i}^{(\rho)}$ by using a strictly  proper scoring rule $R$ and other reviewers' reported reviews as observed outcomes:

\begin{equation}
\label{eq:truth_telling_score_mult_obs}
s_i = \sum_{j \neq i} \left( \gamma R\left(\mathbf{\Phi}_{i}^{(\rho)}, G(\mathbf{r}_{j}) \right) + \lambda \right)
\end{equation}

\noindent where $G$ is some function used by the center to summarize each reviewer $j$'s reported review in a single number, and whose image is equal to the set $\{0, \dots, v\}$. For example, $G$ can be a function that returns the median or the mode of the reported evaluation scores. Honest reporting, \textit{i.e.}, $\mathbf{r}_i = \mathbf{t}_i$,  maximizes reviewers' expected review scores under this setting.

\begin{proposition}
When observing and reporting multiple signals, each reviewer $i \in N$ maximizes his expected review score when $\mathbf{r}_i = \mathbf{t}_i$.
\end{proposition}

\begin{proof}
Due to the autonomy assumption, we restrict ourselves to show that each reviewer $i \in N$ maximizes $\mathbb{E}_{\mathbf{\Theta}_i^{(\rho)}} \left[\gamma R\left(\mathbf{\Phi}_i^{(\rho)}, G(\mathbf{r}_{j})\right)+\lambda\right]$, for $j\neq i$, when $\mathbf{r}_i = \mathbf{t}_i$. Given that $R$ is a strictly proper scoring rule, from Proposition 1 we have that:

\begin{equation*}
\argmax_{\mathbf{\Phi}_i^{(\rho)}} \mathbb{E}_{\mathbf{\Theta}_i^{(\rho)}} \left[\gamma R\left(\mathbf{\Phi}_i^{(\rho)}, G(\mathbf{r}_{j}) \right) + \lambda\right] = \mathbf{\Theta}_i^{(\rho)}
\end{equation*}

When $\mathbf{r}_i = \mathbf{t}_i$, we have by construction that $\mathbf{\Phi}_i^{(\rho)} = \mathbf{\Theta}_i^{(\rho)} $ (see equations (\ref{eq:pos_pred_dist_mult_obs}) and (\ref{eq:est_pos_pred_dist_mult_obs})). Thus, reviewers' expected review scores are maximized when reviewers report honestly. 
\end{proof}

When observing and reporting multiple evaluation scores, a reviewer can weakly maximize his expected review score by reporting a review different than his true review as long as the estimated posterior predictive distributions are the same. For example, when reviewer $i$ reports $\mathbf{r}_i = (1,2,3)$, the resulting estimated posterior predictive distribution is the same as when he reports $\mathbf{r}_i = (3,1,2)$, and, consequently, reviewer $i$ receives the same review score in both cases. This implies that the scoring method in (\ref{eq:truth_telling_score_mult_obs}) is more suitable to a peer-review process where all criteria are equally weighted since honest reporting weakly maximizes expected review scores.

\subsection{Summarizing Signals when Prior Distributions are Non-Informative}

When reviewers report multiple evaluation scores, the intuitive interpretation of review scores as rewards for agreements that arises when using symmetric and bounded strictly proper scoring rules (see Section 4.3) is lost because the elements of the estimated posterior predictive distribution in (\ref{eq:est_pos_pred_dist_mult_obs}) can take on more than two different values.

A different approach that preserves the aforementioned intuitive interpretation when reviewers' prior distributions are non-informative is to ask the reviewers to summarize their observed signals into a single value before reporting it, instead of the center doing it on their behalf. Hence, each reviewer $i$ is now reporting honestly when $r_i = G(\mathbf{t}_i)$, where $G$ is some function suggested by the center whose image is equal to the set $\{0, \dots, v\}$. This new model can be interpreted as if the reviewers were reviewing the manuscript under several criteria and reporting the manuscript's overall evaluation score by reporting the value $G(\mathbf{t}_i)$.

Since reviewers are reporting only one value, we can use the original scoring method in (\ref{eq:truth_telling_score}) to promote honest reporting. We prove below that for any symmetric and bounded strictly proper scoring rule, honest reporting strictly maximizes reviewers' expected review scores under the scoring method in (\ref{eq:truth_telling_score}) if and only if $G$ is the mode of the observed signals, \textit{i.e.}, when $G(\mathbf{t}_i) = \argmax_{x \in \{0, \dots, v\}}  \sum_{k = 1}^{\rho} H(x, t_{i,k})$. Ties between observed signals are broken randomly.

\begin{proposition}
When observing multiple signals and reporting $r_i = G(\mathbf{t}_i)$, each reviewer $i \in N$ with non-informative prior strictly maximizes his expected review score under the scoring method in (\ref{eq:truth_telling_score}), for a symmetric and bounded strictly proper scoring rule $R$, if and only if $r_i = \argmax_{x \in \{0, \dots, v\}}  \sum_{k = 1}^{\rho} H(x, t_{i,k})$.
\end{proposition}

\begin{proof}
Recall that since each reviewer $i \in N$ observes multiple signals, his true posterior predictive distribution is equal to $\mathbb{E}\left[\boldsymbol\omega |  \boldsymbol{\alpha}, \mathbf{t}_{i}\right] = \mathbf{\Theta}_{i}^{(\rho)} = (\theta_0, \theta_1, \dots, \theta_v)$ as shown in (\ref{eq:pos_pred_dist_mult_obs}). Due to Proposition 1 and since reviewers cannot affect their peers' reviews because of the autonomy assumption, we  restrict ourselves to show that each reviewer $i \in N$ maximizes $\mathbb{E}_{\mathbf{\Theta}_i^{(\rho)}} \left[ R\left(\mathbf{\Phi}_i^{(1)}, r_j\right)\right]$, for $j\neq i$, if and only if $r_i = \argmax_{x \in \{0, \dots, v\}}  \sum_{k = 1}^{\rho} H(x, t_{i,k})$, where $\mathbf{\Phi}_i^{(1)} = \mathbb{E}\left[\boldsymbol\omega |  \boldsymbol{\alpha}, r_{i} = G(\mathbf{t}_i)\right]$.

 Let $z \in \{0, \dots, v\}$ be the most common signal observed by reviewer $i$. Hence, reviewer $i$'s subjective probability associated with $z$ is greater than his subjective probability associated with any other signal $y \in \{0, \dots, v\}$, \textit{i.e.}, $\theta_{i, z} > \theta_{i, y}$.

\textbf{(If part)}
Given that $R$ is a symmetric and strictly proper scoring rule and that each reviewer is reporting only one evaluation score, the resulting score from $ R\left(\mathbf{\Phi}_i^{(1)}, r_j\right)$ can take on only two possible values: $\delta_{max}$, if $r_i = r_j$, and $\delta_{min}$ otherwise (see discussion in Section 4.3). When  reporting $r_i = \argmax_{x \in \{0, \dots, v\}}  \sum_{k = 1}^{\rho} H(x, t_{i,k}) = z$, reviewer $i$'s expected review score is $\theta_{i, z} \delta_{max} + \sum_{y \neq z}\theta_{i, y}\delta_{min}$. Given that $\theta_{i, z} > \theta_{i, y},$ for any $y \neq z$, and $\delta_{max} > \delta_{min}$, this expected review score is maximized. Thus, reporting $r_i = \argmax_{x \in \{0, \dots, v\}}  \sum_{k = 1}^{\rho} H(x, t_{i,k}) = z$ maximizes reviewer $i$'s expected review score.

\textbf{(Only-if part)}
Recall that all the elements making up the hyperparameter $\boldsymbol{\alpha}$ have the same value because reviewers' prior distributions are non-informative. Let $\mathbf{\Phi}_i^{(1)}  = (\phi_{i,0}, \dots, \phi_{i,v})$ be reviewer $i$'s estimated posterior predictive distribution computed according to the original scoring method in (\ref{eq:truth_telling_score}) when reviewer $i$ is reporting $r_i = \argmax_{x \in \{0, \dots, v\}}  \sum_{k = 1}^{\rho} H(x, t_{i,k}) = z$, \textit{i.e.}:

\begin{equation*}
\phi_{i,k} = \left\{ \begin{array}{ll}
\frac{\alpha_k +1}{1 + \sum_{x=0}^v \alpha_x} = \frac{\alpha_k +1}{(v+1)\cdot\alpha_k +1} & \textrm{if $k  = z$},\\
\frac{\alpha_k}{1+ \sum_{x=0}^v \alpha_x} = \frac{\alpha_k}{(v+1)\cdot\alpha_k +1} & \textrm{otherwise}.\\
\end{array} \right.
\end{equation*}

For contradiction's sake, suppose that reviewer $i$ maximizes his expected review score by misreporting his review and reporting $r_i = y \neq z$. Let $\tilde{\mathbf{\Phi}}_i^{(1)} = (\tilde{\phi}_{i,0}, \dots, \tilde{\phi}_{i,v})$ be reviewer $i$'s estimated posterior predictive distribution when he is misreporting his review, \textit{i.e.}:

\begin{equation*}
\tilde{\phi}_{i,k} = \left\{ \begin{array}{ll}
\frac{\alpha_k +1}{1 + \sum_{x=0}^v \alpha_x} = \frac{\alpha_k +1}{(v+1)\cdot\alpha_k +1} & \textrm{if $k = y$},\\
\frac{\alpha_k}{1+ \sum_{x=0}^v \alpha_x} = \frac{\alpha_k }{(v+1)\cdot\alpha_k +1} & \textrm{otherwise}.\\
\end{array} \right.
\end{equation*}

As discussed in Section 4.3, the term $R\left(\mathbf{\Phi}_i^{(1)}, r_j \right)$ can take on only two possible values whenever $R$ is a symmetric scoring rule. Consequently, $R\left(\tilde{\mathbf{\Phi}}_i^{(1)}, k\right) = R\left(\mathbf{\Phi}_i^{(1)}, k\right)$ for $k \neq z, y$.  A consequence of our assumption that reviewer $i$ maximizes his expected review score by misreporting his review is that $\mathbb{E}_{\mathbf{\Theta}_i^{(\rho)}} \left[R\left(\tilde{\mathbf{\Phi}}_i^{(1)}, r_j\right)\right] \geq
\mathbb{E}_{\mathbf{\Theta}_i^{(\rho)}} \left[R\left(\mathbf{\Phi}_i^{(1)}, r_j\right)\right]$. Assuming that $R$ is a symmetric and bounded proper scoring rule, this inequality becomes:

\begin{align*}
\sum_{k = 0}^v \theta_{i, k} R\left(\tilde{\mathbf{\Phi}}_i^{(1)}, k\right)
&\geq
\sum_{k = 0}^v \theta_{i, k} R\left(\mathbf{\Phi}_i^{(1)}, k\right) & \implies \label{eq:prop_3}\\
\theta_{i, z} R\left(\tilde{\mathbf{\Phi}}_i^{(1)}, z\right) + \theta_{i, y} R\left(\tilde{\mathbf{\Phi}}_i^{(1)}, y\right)  
&\geq
\theta_{i, z} R\left(\mathbf{\Phi}_i^{(1)}, z\right) + \theta_{i, y}R\left(\mathbf{\Phi}_i^{(1)}, y\right) & \implies \nonumber\\
\theta_{i, y} 
&\geq
\theta_{i, z} \left(\frac{ R\left(\mathbf{\Phi}_i^{(1)}, z\right) - R\left(\tilde{\mathbf{\Phi}}_i^{(1)}, z\right)} {R\left(\tilde{\mathbf{\Phi}}_i^{(1)}, y\right) - R\left(\mathbf{\Phi}_i^{(1)}, y\right)} \right) \nonumber
\end{align*}

The second line follows from the fact that $R\left(\tilde{\mathbf{\Phi}}_i^{(1)}, k\right) = R\left(\mathbf{\Phi}_i^{(1)}, k\right)$ for $k \neq z, y$. Regarding the last line, we have by construction that $R\left(\mathbf{\Phi}_i^{(1)}, z\right) = R\left(\tilde{\mathbf{\Phi}}_i^{(1)}, y\right) = \delta_{max}$, and $R\left(\tilde{\mathbf{\Phi}}_i^{(1)}, z\right) = R\left(\mathbf{\Phi}_i^{(1)}, y\right) = \delta_{min}$. Consequently, we obtain that $\theta_{i, y} \geq \theta_{i, z}$. As we stated before, since $z$ is the most common signal observed by reviewer $i$, then $\theta_{i, z} > \theta_{i, y}$. Thus, we have a contradiction. So, $\mathbb{E}_{\mathbf{\Theta}_i^{(\rho)}}\left[R\left(\tilde{\mathbf{\Phi}}_i^{(1)}, r_j\right)\right] < \mathbb{E}_{\mathbf{\Theta}_i^{(\rho)}} \left[R\left(       \mathbf{\Phi}_i^{(1)},  r_j\right) \right]$, \textit{i.e.}, reviewer $i$ maximizes his expected review score only if he reports $r_i = \argmax_{x \in \{0, \dots, v\}}  \sum_{k = 1}^{\rho} H(x, t_{i,k}) = z$. 
\end{proof}

In other words, the above proposition says that each reviewer should report the evaluation score most likely to be deserved by the manuscript when their prior distributions are non-informative and they are rewarded according to the scoring method in (\ref{eq:truth_telling_score}). Any other evaluation score has a lower associated subjective probability and, consequently, reporting it results in a lower expected review score. To summarize, Proposition 4 implies that the scoring method proposed in (\ref{eq:truth_telling_score}) induces honest reporting by rewarding agreements whenever reviewers' prior distributions are non-informative and the center is interested in the mode of each reviewer's observed signals. It is noteworthy that Proposition 4 does not assume that the center knows \textit{a priori} the number of observed signals ${\rho}$, thus providing more flexibility for practical applications of our method.

\section{Finding a Consensual Review}

After reviewers report their reviews and receive their review scores, there is still the question of how the center will use the reported reviews in making a suitable decision. Since reviewers are not always in agreement, belief aggregation methods must be used to combine the reported reviews into a single representative review. The traditional average method is not necessarily the best approach since unreliable reviewers might have a big impact on the aggregate review. Moreover, a consensual review is desirable because it represents a review that is acceptable by all.

In this section, we propose an adaptation of a classical mathematical method to find a consensual review. Intuitively, it works as if reviewers were constantly updating their reviews in order to aggregate knowledge from others. The scoring concepts introduced in previous sections are incorporated by the reviewers when updating their reviews. In what follows,  for the sake of generality, we assume that reviewers evaluate the manuscript  under $\rho \in \mathbb{N}^+$ criteria, \textit{i.e.}, each reviewer $i$ observes $\rho$  signals from the underlying distribution that represents the quality of the manuscript and report a vector $\mathbf{r}_i = (r_{i,1}, \dots, r_{i,\rho})$ of evaluation scores,  where $r_{i,k} \in \{0, \dots, v\}$ for all $k$. The center then estimates reviewers $i$'s posterior predictive distribution $\mathbb{E}\left[\boldsymbol\omega |  \boldsymbol{\alpha}, \mathbf{r}_{i}\right]$, referred to as $\mathbf{\Phi}_{i}^{(\rho)}$ for ease of exposition, as in (\ref{eq:est_pos_pred_dist_mult_obs}). We relax our basic model by allowing the evaluation scores in the aggregate review to take on any real value between $0$ and the best evaluation score $v$.

\subsection{DeGroot's Model}

\cite{DeGroot:1974} proposed a model that describes how a group might reach a consensus by pooling their individual opinions. When applying this model to a peer-review setting, each reviewer $i$ is first informed of others' reported reviews. In order to accommodate the information and expertise of the rest of the group, reviewer $i$ then updates his own review as follows:

\begin{equation*}
\mathbf{r}_i^{(1)} = \sum_{j=1}^n w_{i, j } \mathbf{r}_j
\end{equation*}

\noindent where $w_{i,j}$ is a weight that reviewer $i$ assigns to reviewer $j$'s reported review when he carries out this update. It is assumed that $w_{i,j} \geq 0$, for every reviewer $i$ and $j$, and $\sum_{j=1}^n w_{i,j} = 1$. In this way, each updated review takes the form of a linear combination of reported reviews, also known as a \emph{linear opinion pool}. The weights must be chosen on the basis of the relative importance that reviewers assign to their peers' reviews. The whole updating process can be written in a  more general form using matrix notation: $\mathbf{R}^{(1)} = \mathbf{W}\mathbf{R}^{(0)}$, where:

\begin{equation*}
\mathbf{W} = \left[\begin{array}{cccc}
w_{1,1} & w_{1, 2} & \cdots & w_{1, n} \\
w_{2,1} & w_{2, 2} & \cdots & w_{2, n} \\
\vdots  & \vdots  & \ddots & \vdots  \\
w_{n,1} & w_{n, 2} & \cdots & w_{n, n} \\
\end{array}
\right] \quad  \mbox{and} \quad
\mathbf{R}^{(0)} 
=
\left[\begin{array}{c}
\mathbf{r}_{1} \\
\mathbf{r}_{2} \\
\vdots\\
\mathbf{r}_{n} \\
\end{array} 
\right]
=
\left[\begin{array}{cccc}
r_{1,1}  & r_{1, 2} & \cdots & r_{1,\rho} \\
r_{2,1}  & r_{2, 2} & \cdots & r_{2, \rho} \\
\vdots  & \vdots  & \ddots & \vdots  \\
r_{n,1}  & r_{n, 2} & \cdots & r_{n, \rho} \\
\end{array}
\right]
\end{equation*}

Since all the original reviews have changed, the reviewers might wish to update their new reviews in the same way as they did before. If there is no basis for the reviewers to change their assigned weights, the whole updating process after $t$ revisions can then be represented as follows:

\begin{equation}
\label{eq:DeGroot_consensus}
\mathbf{R}^{(t)} = \mathbf{W}\mathbf{R}^{(t-1)} = \mathbf{W}^{t}\mathbf{R}^{(0)}
\end{equation}

Let $\mathbf{r}^{(t)}_i = \left(r^{(t)}_{i,1}, \dots, r^{(t)}_{i,\rho}\right)$ be reviewer $i$'s review after $t$ revisions, \textit{i.e.}, it denotes the $i$th row of the matrix $\mathbf{R}^{(t)}$. We say that a \emph{consensus} is reached if and only if $\mathbf{r}^{(t)}_i = \mathbf{r}^{(t)}_j$, for every reviewer $i$ and $j$, when $t \rightarrow \infty$.

\subsection{Review Scores as Weights}

The original method proposed by \cite{DeGroot:1974} does not encourage honesty in a sense that reviewers can assign weights to their peers' reviews however they wish so as long as the weights are consistent with the construction previously defined. Furthermore, it requires the disclosure of reported reviews to the whole group when  reviewers are weighting others' reviews, a fact which might be troublesome when the reviews are of a sensitive nature.

A possible way to circumvent the aforementioned problems is to derive weights from the original reported reviews by taking into account review scores. In particular, we assume the weight that a reviewer assigns to a peer's review is directly related to how close their estimated posterior predictive distributions are, where closeness is defined by an underlying proper scoring rule. We provide behavioral foundations for such an assumption in the following subsection. Formally, the weight that reviewer $i$ assigns to reviewer $j$'s reported review is  computed as follows:

\begin{equation}
\label{eq:consensus_weight}
w_{i, j} = \frac{\mathbb{E}_{\mathbf{\Phi}_i^{(\rho)}} \left[ \gamma R\left(\mathbf{\Phi}_j^{(\rho)}, e\right) + \lambda \right] }{ \sum_{j=1}^n  \mathbb{E}_{\mathbf{\Phi}_i^{(\rho)}} \left[ \gamma R\left(\mathbf{\Phi}^{(\rho)}_j, e\right) + \lambda \right]}
\end{equation}

\noindent that is, the weight $w_{i, j}$ is proportional to the expected review score that reviewer $i$ would receive if he had reported the review reported by reviewer $j$, where the expectation is taken with respect to reviewer $i$'s estimated posterior predictive distribution. Consequently, the weight that each reviewer indirectly assigns to his own review is always the highest because $R$ is a strictly proper scoring rule, \textit{i.e.}, $\argmax_{\mathbf{\Phi}_j^{(\rho)}} \mathbb{E}_{\mathbf{\Phi}_i^{(\rho)}} \left[ \gamma R\left(\mathbf{\Phi}_j^{(\rho)}, e\right) + \lambda \right] = \mathbf{\Phi}_i^{(\rho)}$. We assume that $\mathbb{E}_{\mathbf{\Phi}_i^{(\rho)}} \left[ \gamma R\left(\mathbf{\Phi}_j^{(\rho)}, e\right) + \lambda \right] > 0$, for every reviewer $i$ and $j$. As long as $R$ is bounded, this assumption can be met by appropriately setting the value of $\lambda$. Consequently,  $0 < w_{i, j} <1$, for every $i, j \in N$. Moreover, $\sum_{j=1}^n w_{i, j} =1$ because the denominator of the fraction in (\ref{eq:consensus_weight}) normalizes the weights so they sum to one.

In the interest of reaching a consensus, DeGroot's method in (\ref{eq:DeGroot_consensus}) is applied to the original reported reviews using the weights as defined in (\ref{eq:consensus_weight}). We show that a consensus is always reached under this proposed method whenever the review scores are positive.

\begin{proposition}
If $\mathbb{E}_{\mathbf{\Phi}_i^{(\rho)}} \left[ \gamma R\left(\mathbf{\Phi}_j^{(\rho)}, e\right) + \lambda \right] > 0$, for every reviewer $i, j \in N$, then $\mathbf{r}_{i}^{(t)} = \mathbf{r}_{j}^{(t)}$ when $t \rightarrow \infty$.
\end{proposition}


\begin{proof}
Due to the assumption that $\mathbb{E}_{\mathbf{\Phi}_i^{(\rho)}} \left[ \gamma R\left(\mathbf{\Phi}_j^{(\rho)}, e\right) + \lambda \right] > 0$, for every reviewer $i$ and $j$, all the elements of the matrix $\mathbf{W}$ in (\ref{eq:DeGroot_consensus}) are strictly greater than zero and strictly less than one. Moreover, the sum of the elements in any row is equal to one. Consequently, $W$ can be regarded as a $n\times n$ stochastic matrix, or a one-step transition probability matrix of a Markov chain with $n$ states and stationary probabilities. Furthermore, the underlying Markov chain is aperiodic and irreducible. Therefore, a standard limit theorem of Markov chains applies in this setting, namely given an  aperiodic and irreducible Markov chain with transition probability matrix $\mathbf{W}$, every row of the matrix $ \mathbf{W}^{t}$ converges to the same probability vector when $t \rightarrow\infty$ \citep{Ross:1995}. 
\end{proof}

Recall that $\mathbf{r}_i^{(t)} = \sum_{j=1}^n w_{i,j}\mathbf{r}_j^{(t-1)} = \sum_{j=1}^n w_{i,j}\sum_{k=1}^n w_{j,k} \mathbf{r}_k^{(t-2)} = \dots =  \sum_{j=1}^n \beta_j\mathbf{r}_j^{(0)}$, where $\mathbf{\beta} = (\beta_1, \beta_2, \dots, \beta_n)$ is a probability vector that incorporates all the previous weights. This equality implies that the consensual review can be represented as an instance of the linear opinion pool. Hence, an interpretation of the proposed method is that reviewers reach a consensus regarding the weights in (\ref{eq:DeGroot_consensus}). When $\mathbf{\beta} = (1/n, 1/n, \dots, 1/n)$ in the above equality, the underlying linear opinion pool becomes the average of the reported evaluation scores. A drawback with an averaging approach is that it does not take into account the scoring concepts introduced in the previous sections, a fact which might favor unreliable reviewers. Moreover, disparate reviews might have a big impact on the resulting aggregate review. On the other hand, under our approach to find $\mathbf{\beta}$, reviewers weight down reviews far from their own reviews, which implies that the proposed method might be less influenced by disparate reviews. A numerical example in subsection 6.4 illustrates this point. The experimental results discussed in Section 7 show that our method to find a consensual review is consistently more accurate than the traditional average method.

\subsection{Behavioral Foundation}

The major assumption regarding our method for finding a consensual review is that reviewers assign weights according to (\ref{eq:consensus_weight}). An interesting interpretation of (\ref{eq:consensus_weight}) arises when the proper scoring rule $R$ is \emph{effective} with respect to a metric $M$. Formally, given a metric $M$ that assigns a real number to any pair of probability vectors, which can be seen as the  shortest distance between the two probability vectors, we say that a scoring rule $R$ is effective with respect to $M$ if the following relation holds for all probability vectors $\mathbf{\Phi}_i^{(\rho)}$, $\mathbf{\Phi}_j^{(\rho)}$, and $\mathbf{\Phi}_k^{(\rho)}$ \citep{Friedman:1983}:

\begin{equation*}
M\left(\mathbf{\Phi}_i^{(\rho)}, \mathbf{\Phi}_j^{(\rho)}\right) < M\left(\mathbf{\Phi}_i^{(\rho)}, \mathbf{\Phi}_k^{(\rho)}\right)
\iff
\mathbb{E}_{\mathbf{\Phi}_i^{(\rho)}} \left[\gamma R\left(\mathbf{\Phi}_j^{(\rho)}, e\right) + \lambda \right] > \mathbb{E}_{\mathbf{\Phi}_i^{(\rho)}} \left[\gamma R\left(\mathbf{\Phi}_k^{(\rho)}, e\right) + \lambda \right]
\end{equation*}

Thus, when $R$ is effective with respect to a metric $M$, the higher the weight one reviewer assigns to a peer's review in (\ref{eq:consensus_weight}), the closer their estimated posterior predictive distributions are according to the metric $M$. In other words, when using effective scoring rules, reviewers naturally prefer reviews close to their own reported reviews, and the weight that each reviewer assigns to his own review is always the highest one. Hence, in spirit, the resulting learning model in (\ref{eq:DeGroot_consensus}) can be seen as a model of \emph{anchoring} \citep{Tversky:1974} in a sense that the review of a reviewer is an ``anchor", and subsequent updates are biased towards reviews close to the anchor.

\cite{Friedman:1983} discussed some examples of effective scoring rules. For example, the quadratic scoring rule in (\ref{eq:quad_scor_rul}) is effective with respect to the root-mean-square deviation, the spherical scoring rule is effective with respect to a renormalized $L_2$-metric, whereas the logarithmic scoring rule is not effective with respect to any metric \citep{Nau:1985}.

\subsection{Numerical Example}

Consider a peer-review process where the best evaluation score is four ($v=4$), three reviewers $(n=3)$ observe three signals ($\rho = 3$), and they report the following reviews:  $\mathbf{r}_{1} = (0, 1, 3)$, $\mathbf{r}_{2} = (0, 2, 3)$, and $\mathbf{r}_3 = (4, 4, 4)$. Consequently,  the matrix $\mathbf{R}^{(0)}$ in (\ref{eq:DeGroot_consensus}) is:

\begin{equation*}
\mathbf{R}^{(0)} = \left[\begin{array}{ccc}
0 & 1 & 3\\
0 & 2 & 3  \\
4 & 4 & 4 \\
\end{array}
\right]
\end{equation*}

Consider the hyperparameter $\boldsymbol\alpha = (1,1,1,1,1)$. Hence, the estimated posterior predictive distributions are $\mathbf{\Phi}_1^{(3)} = (2/8, 2/8, 1/8, 2/8, 1/8)$, $\mathbf{\Phi}_2^{(3)} = (2/8, 1/8, 2/8, 2/8, 1/8)$, and $\mathbf{\Phi}_3^{(3)} = (1/8, 1/8, 1/8, 1/8, 4/8)$. Assume that $R$ in (\ref{eq:consensus_weight}) is the quadratic scoring rule in (\ref{eq:quad_scor_rul}), and let $\gamma = 1$ and $\lambda = 1$ in order for the resulting expected values in (\ref{eq:consensus_weight}) to be always positive. We obtain:

\begin{align*}
\mathbf{W} = \left[\begin{array}{ccc}
0.3545 & 0.3455 & 0.3000 \\
0.3455 & 0.3545 & 0.3000 \\
0.3158 & 0.3158	& 0.3684 \\
\end{array}
\right]
\end{align*}

Focusing on the main diagonal of $\mathbf{W}$, we notice that each reviewer always assigns the highest weight to his own review. From the first row of $\mathbf{W}$, we can see that reviewer $1$ assigns a high weight to his review and to reviewer $2$'s review, and a lower weight to reviewer $3$'s review. This happens because reviewer 3's review is very distant from the others' reported reviews. We can draw similar conclusions from the other rows. We then obtain the following weights when carrying out DeGroot's method with the weights calculated according to (\ref{eq:consensus_weight}):

\begin{equation*}
\lim_{t\rightarrow \infty}\mathbf{W}^{t} = \left[\begin{array}{ccc}
0.3390 & 0.3390 & 0.3220 \\
0.3390 & 0.3390 & 0.3220 \\
0.3390 & 0.3390 & 0.3220 \\
\end{array}
\right]
\end{equation*}

\noindent and, consequently, the consensual review is represented by any row of the matrix $\lim_{t\rightarrow \infty}\mathbf{W}^{t}\mathbf{R}^{(0)}$, which results in the vector $(1.288, 2.305, 3.322)$. It is worthwhile to discuss an interesting point regarding the above example. The aggregate review would be $(1.333, 2.333, 3.333)$ if it was equal to the average of the reported reviews. Hence, reviewer $3$'s review would have more impact on the aggregate review because the evaluation scores in the average review are all greater than the corresponding evaluation scores in the consensual review. In our proposed method, the influence of reviewer $3$ on the aggregate review is diluted because his review is very different from the others' reviews. More formally, reviewer $3$'s estimated posterior predictive distribution is very distant from the others' estimated posterior predictive distributions when measured according to the root-mean-square deviation, the metric associated with the quadratic scoring rule.

\section{Experiments}

In this section, we describe a peer-review experiment designed to test the efficacy of both the proposed scoring method and the proposed aggregation method. In the following subsections, we discuss Amazon's Mechanical Turk, the platform used in our experiments, the experimental design, and our results.

\subsection{Amazon's Mechanical Turk}

Amazon's Mechanical Turk\footnote{http://www.mturk.com} (AMT) is an online labor market originally developed for human computation tasks, \textit{i.e.}, tasks that are relatively easy for human beings, but nonetheless challenging or even currently impossible for computers, \textit{e.g.}, audio transcription, filtering adult content, extracting data from images, \textit{etc}. Several studies have shown that AMT can effectively be used as a means of collecting valid data in these settings \citep{Snow:2008, Marge:2010}.

More recently, AMT has also been used as a platform for conducting behavioral experiments \citep{Mason:2012}. One of the advantages that it offers to researchers is the access to a large, diverse, and stable pool of people willing to participate in the experiments for relatively low pay, thus simplifying the recruitment process and allowing a faster iteration between developing theory and executing experiments. Furthermore, AMT provides an easy-to-use built-in mechanism to pay workers that greatly reduces the difficulties of compensating individuals for their participation in the experiments, and a built-in reputation system  that helps requesters distinguish between good and bad workers and, consequently, to ensure data quality. Numerous studies have shown that results of behavioral studies conducted on AMT are comparable to results obtained in other online domains as well as in offline settings \citep{Buhrmester:2011, Horton:2011}, thus providing evidence that AMT is a valid means of collecting behavioral data.

\subsection{Experimental Design}

We designed a task on AMT that required workers, henceforth referred to as reviewers, to review 3 short texts under three different criteria: \emph{Grammar}, \emph{Clarity}, and \emph{Relevance}. The first two texts were extracts from published poems, but with some original words intentionally replaced by misspelled words. The third text contained random words presented in a semi-structured way. All the details regarding the texts are included in the appendix. For each text, three questions were presented to the reviewers, each one having three possible responses ordered in decreasing negativity order:

\begin{itemize}
\item Grammar: does the text contain misspellings, syntax errors, \textit{etc}.?

	\begin{itemize}
	\item A lot of grammar mistakes
	\item A few grammar mistakes
	\item No grammar mistakes
	\end{itemize}

\item Clarity: does the text, as a whole, make any sense?

	\begin{itemize}
	\item The text does not make sense
	\item The text makes some sense
	\item The text makes perfect sense
	\end{itemize}

\item Relevance: could the text be part of a poem related to love?

	\begin{itemize}
	\item The text cannot be part of a love poem
	\item The text might be part of a love poem
	\item The text is definitely part of a love poem
	\end{itemize}

\end{itemize}

Words with subjective meaning were intentionally used so as to simulate the subjective nature of the evaluation scores in a review, \textit{e.g.}, ``a lot", ``a few", \textit{etc.} Each individual response was translated into an evaluation score inside the set $\{0, 1 ,2\}$. The most negative response received the score $0$, the middle response received the score $1$, and the most positive response received the score $2$. Thus, each reviewer reported a vector of 9 evaluation scores (3 texts times 3 criteria).

We recruited 150 reviewers on AMT, all of them residing in the United States of America and older than 18 years old. They were required to accomplish the task in at most 20 minutes. Reviewers were split into 3 groups of equal size. After accomplishing the task, every reviewer in every group received a payment of $20$ cents. A study done by \cite{Ipeirotis:2010} showed that more than 90$\%$ of the tasks on AMT have a baseline payment less than $\$0.10$, and $70\%$ of the tasks have a baseline payment less than $\$0.05$. Thus, our baseline payment was much higher than the average payment from other jobs posted to the AMT marketplace.

Each reviewer was randomly assigned into one of the three groups. Reviewers in two of the groups, the treatment groups, could earn an additional \emph{bonus} of up to $10$ cents. Reviewers in the first treatment group, referred to as the \emph{Bonus Group} (BG), were informed that their bonuses would be proportional to the number of reviews similar to their reported reviews. Reviewers in the second treatment group, the \emph{Bonus and Information Group} (BIG), received similar information, but they also received a short summary of some theoretical results presented in this paper:

\begin{quote}
``A group of researchers from the University of Waterloo (Canada) formally showed that the best strategy to maximize your expected bonus in this setting is by being honest, \textit{i.e.}, by considering each question thoroughly and deciding the best answers according to your personal opinion".
\end{quote}

Members of the third group, the \emph{Control Group} (CG), neither received extra explanations nor bonuses. Their reported reviews were used as the control condition. Bonuses were computed  by rewarding agreements as described in Section 4.3. Due to the one-shot nature of this peer-review task, we assumed that Dirichlet priors were non-informative with hyperparameter $\boldsymbol\alpha = (1,1,1)$. For each reported evaluation score, there could be at most 49 similar reported evaluation scores because each group had 50 members. We then used the formula $\frac{10}{9}\times \frac{\# \mbox{agreements}}{49}$ to calculate the reward for an individual evaluation score. Given that each reviewer reported 9 evaluation scores, if the evaluation scores reported by all members of a group were the same, then all group members would received the maximum bonus of 10 cents. The provided bonuses can be seen as review scores. Our primary objective when performing this experiment was to empirically investigate the extent to which providing review scores affects the quality of the reported reviews.

\subsection{Gold-Standard Evaluation Scores}

Since the source and original content of each text were known \textit{a priori}, \textit{i.e.}, before the experiments were conducted, we were able to derive \emph{gold-standard reviews} for each text. In order to avoid confirmation bias\footnote{The tendency to interpret information in a way that confirms one's preconceptions \citep{Plous:1993}.}, we asked five professors and tutors from the English and Literature Department at the University of Waterloo to provide their reviews for each text. We set the \emph{gold-standard evaluation score} for each criterion in a text as the median of the evaluation scores reported by the professors and tutors. Coincidentally, each median value was also the mode of the underlying evaluation scores. All the evaluation scores reported by the professors and tutors as well as the respective gold-standard evaluation scores are in the appendix.

\subsection{Hypotheses}

Our first research question was whether or not providing review scores through pairwise comparisons makes the reported reviews more accurate, \textit{i.e.}, closer to the gold-standard reviews. Based on our theoretical results, our hypothesis was:

\begin{hypothesis} 
The average accuracy of group BIG is greater than the average accuracy of group BG, which in turn is greater than the average accuracy of group CG.
\end{hypothesis}

In other words, the resulting reviews would be on average more accurate when reviewers received review scores, and the extra explanation regrading the theory behind the scoring method would provide more credibility to it, thus making the reviews more accurate. Regarding the resulting bonuses, since honest reporting maximizes reviewers' expected review scores in our model, our second hypothesis was:

\begin{hypothesis} 
The average bonus received by members of group BIG is greater than the average  bonus received by members of group BG, which in turn is greater than the average  bonus received by members of group CG.
\end{hypothesis}

In order to test whether or not Hypothesis 2 was true, we used the bonus the members of group CG would have received had they received any bonus. It is important to note that Hypothesis 1 was measured by comparing how close the reported reviews were to the gold-standard reviews, whereas Hypothesis 2 was measured by making pairwise comparisons between reported reviews: the higher the number of agreements, the greater the resulting bonus.

Another metric used to compare groups' performance was the \emph{task completion time}. The amount of time spent by reviewers on the reviewing task can be seen as a proxy for the effort they exerted to complete the task. Regarding this metric, we expected reviewers who received review scores to be more cautious when completing their tasks. Moreover, the extra explanation regrading the theory behind the scoring method would provide more credibility to it, thus making the members of group BIG work harder on the task. Hence, our third hypothesis was:

\begin{hypothesis} 
The average task completion time of group BIG is greater than the average task completion time of group BG, which in turn is greater than the average task completion time of group CG.
\end{hypothesis}

Finally, we believed that the consensual review, computed as described in Section 6, would be more accurate than the average review since disparate  reviews are less likely to have a big influence on the consensual review than on the average review. Hence, our fourth hypothesis was:

\begin{hypothesis} 
The average accuracy of the consensual review is greater than the average accuracy of the average review.
\end{hypothesis}

\subsection{Experimental Results}

\subsubsection{Accuracy on Individual Criteria}

In our first analysis, we computed the absolute difference between each reported evaluation score and the corresponding gold-standard evaluation score. Thus, the outcome measure was an integer with a value between zero and two, and the closer this value was to zero, the better the resulting accuracy. Table~\ref{tab:mean_values_individual} shows the average accuracy of each group on individual criteria.

\begin{table}[t]
\caption{Accuracy of each group on individual criteria. The average of the absolute difference between the reported evaluation scores and the corresponding gold-standard evaluation scores is shown below each group. For each criterion, the lowest average is highlighted in bold. The standard deviations are in parenthesis. One-tailed p-values resulting from rank-sum tests are shown in the last three columns. Given the notation A-B, the null hypothesis is that the outcome measures resulting from groups A and B are equivalent, and the alternative hypothesis is that the outcome measure  resulting from group A is less  than the outcome measure  resulting from group B.}
\label{tab:mean_values_individual}
\centering
\begin{tabular}{l l c c c c c c} 
\toprule
 &  &       &       &       &  \multicolumn{3}{c}{$p$-values} \\ \cmidrule(r){6-8}   
 &  & BG & BIG & CG &  BIG-BG &  BIG-CG &  BG-CG\\
\midrule
       & Grammar   & 0.5000   & \textbf{0.3200}   & 0.4400 & 0.035** & 0.110 & 0.726 \\
       &           & (0.5051) & (0.4712) & (0.5014) \\ 
Text 1 & Clarity   & 0.8200   & \textbf{0.6200}   & 0.8600 & 0.065* & 0.052* & 0.413  \\
       &           & (0.6606) & (0.6024) & (0.7287) \\ 
       & Relevance & 0.2200   & \textbf{0.2000}   & 0.3000 & 0.484 & 0.213 & 0.230  \\
       &           & (0.5067) & (0.4518) & (0.5803) \\ 
       \\
       & Grammar   & 0.4400   & \textbf{0.3600}   & 0.3800 & 0.209 & 0.420 & 0.729  \\
       &           & (0.5014) & (0.4849) & (0.4903) \\ 
Text 2 & Clarity   & 0.5000   & \textbf{0.3800}   & 0.5400 & 0.155 & 0.067* & 0.325  \\
       &           & (0.6468) & (0.6024) & (0.6131) \\ 
       & Relevance & \textbf{0.4400}   & 0.6400   & 0.6600 & 0.977 & 0.419 & 0.014**  \\
       &           & (0.5014) & (0.4849) & (0.4785)\\ 
       \\
       & Grammar   & \textbf{0.7600 }  & 0.7800   & 1.0200 & 0.539	& 0.077* & 0.061*  \\
       &           & (0.8466) & (0.8640) & (0.8449)\\ 
Text 3 & Clarity   & 0.1400   & \textbf{0.0000}   & 0.1600 & 0.006** & 0.002** & 0.301 \\
       &           & (0.4046) & (0.0000) & (0.3703)\\
       & Relevance & 0.1200   & \textbf{0.1000}   & 0.2000 & 0.491 & 0.112 & 0.122 \\ 
       &           & (0.4352) & (0.3642) & (0.4949)\\ 
\bottomrule
*  $\mbox{   } p \leq 0.1$\\
** $p \leq 0.05$
\end{tabular}
\end{table}

Focusing first on the groups BG and BIG, the group BIG is the most accurate group on all criteria, except for the criterion Relevance in Text 2 and the criterion Grammar in Text 3. This result is statistically significant with $p\mbox{-value} \leq 0.1$ in three out of the seven cases in which BIG is more accurate than BG. In two out of these three statistically significant cases, this result is also statistically significant with $p\mbox{-value} \leq 0.05$. BG is more accurate than BIG in only two criteria. This result is only statistically significant for the criterion Relevance in Text 2 ($p\mbox{-value} \leq 0.05$).

The group CG, the control condition that involved no incentives beyond the baseline compensation offered for completing the task, never outperforms both BG and BIG at the same time, and it is the less accurate group in seven out of nine criteria. In two (respectively, four) occasions, CG is statistically significantly less accurate than BG (respectively BIG) with $p\mbox{-value} \leq 0.1$.

Giving these results, we conclude that Hypothesis 1 is true for individual criteria, \textit{i.e.}, the resulting reviews are on average more accurate when using review scores, and the extra explanation regrading the theory behind the scoring method seems to provide more credibility to it, thus improving the accuracy of the reported reviews.

\subsubsection{Aggregate Accuracy}

We also computed the aggregate accuracy of each group for each text as well as for the whole task. In the former case, the outcome measure was the sum of the absolute difference between each reported evaluation score for a given text and the corresponding gold-standard evaluation score. For example, given $(0, 1, 2)$ as the reported evaluation scores for Text 1, and $(1, 2, 2)$ as the corresponding gold-standard evaluation scores, the outcome measure for Text 1 would be $|0 - 1| + |1 - 2| + |2 - 2| = 2$. For the whole task, we summed the absolute differences across all criteria and texts. Table~\ref{tab:mean_values_aggregate} shows the aggregate accuracy of each group.

\begin{table} [t]
\caption{Aggregate accuracy of each group. The average of the sum of the absolute difference between the reported evaluation scores and the corresponding gold-standard evaluation scores is shown below each group. For each text and for the whole task, the lowest average is highlighted in bold. The standard deviations are in parenthesis. One-tailed p-values resulting from rank-sum tests are given in the last three columns. Given the notation A-B, the null hypothesis is that the outcome measures resulting from groups A and B are equivalent, and the alternative hypothesis is that the outcome measure resulting from group A is less than the outcome measure resulting from group B.}
\label{tab:mean_values_aggregate}
\centering
\begin{tabular}{l c c c c c c} 
\toprule
 &        &       &       &  \multicolumn{3}{c}{$p$-values} \\ \cmidrule(r){5-7}   
 &  BG & BIG & CG &  BIG-BG &  BIG-CG &  BG-CG\\
\midrule
Text 1 & 1.5400   & \textbf{1.1400}   & 1.6000 & 0.043** & 0.085* & 0.588 \\
       & (1.1287) & (1.0304) & (1.4142)\\ 
Text 2 & \textbf{1.3800}   & \textbf{1.3800}   & 1.5800 & 0.547 & 0.163 & 0.148 \\
       & (1.0669) & (0.9666) & (0.9916)\\ 
Text 3 & 1.0200   & \textbf{0.8800}   & 1.3800 & 0.394 & 0.020** & 0.052* \\
       & (1.1865) & (0.9179) & (1.1933)\\ 
Overall& 3.9400   & \textbf{3.4000}   & 4.5600 & 0.110 & 0.002** & 0.064* \\
       & (2.2352) & (1.6903) & (2.1301)\\ 
\bottomrule
*  $\mbox{   } p \leq 0.1$\\
** $p \leq 0.05$
\end{tabular}
\end{table}

For every single text as well as for the overall task, members of the group CG report less accurate reviews than members of the group BG and the group BIG. For the group BG, this result is statistically significant for Text 3 and for the overall task ($p\mbox{-value} \leq 0.1$). For the group BIG, this result is statistically significant for Text 1 ($p\mbox{-value} \leq 0.1$), Text 3 ($p\mbox{-value} \leq 0.05$), and for the whole task ($p\mbox{-value} \leq 0.05$). Thus, the experimental results suggest  that providing review scores produces a significant improvement in quality over the control condition. Moreover, providing an extra explanation about the theory behind the scoring method improves the final quality of the reviews because, on average, the reviews from group BIG are more accurate than the reviews from group BG. This result is statistically significant for Text 1 ($p\mbox{-value} \leq 0.05$). Therefore, we conclude that Hypothesis 1 is also true on the aggregate level.

\subsubsection{Bonus}

The average bonus per group is shown in the first row of Table~\ref{tab:rewards_time}. From it, we conclude that Hypothesis 2 is true, \textit{i.e.}, the average bonus received by members of BIG is greater than the average bonus received by members of BG, which in turn is greater than the average bonus hypothetically received by members of CG. All these results are statistically significant with $p\mbox{-value} \leq 0.05$. In other words, providing review scores and informing reviewers about the theory behind the scoring method do indeed increase the number of reported reviews that are similar.

\begin{table}
\caption{Average bonus and completion time per group. The highest average values are highlighted in bold. The standard deviations are in parenthesis. One-tailed p-values resulting from rank-sum tests are given in the last three columns. Given the notation A-B, the null hypothesis is that the outcome measures resulting from groups A and B are equivalent, and the alternative hypothesis is that the outcome measure resulting from group A is greater than the  outcome measure resulting from group B.}
\label{tab:rewards_time}
\centering
\begin{tabular}{l c c c c c c} 
\toprule
 &        &       &       &  \multicolumn{3}{c}{$p$-values} \\ \cmidrule(r){5-7}   
 &  BG & BIG & CG &  BIG-BG &  BIG-CG &  BG-CG\\
\midrule
Bonus  & 0.053    & \textbf{0.058}   & 0.050 & $<$ 0.0005** & $<$0.0005** & 0.0025** \\
       & (0.0086) & (0.0073) & (0.0078)\\ 
Time   & 178.66 & \textbf{215.90}   & 196.36 & 0.0232** & 0.0257** & 0.4208\\
       & (87.4495)& (127.7471) & (149.0788)\\ 
\bottomrule
*  $\mbox{   } p \leq 0.1$\\
** $p \leq 0.05$
\end{tabular}
\end{table}

Interestingly, there is a strong negative correlation between bonuses and the aggregate absolute error for the whole task shown in the fourth row of Table~\ref{tab:mean_values_aggregate}, even though the former is computed by making pairwise comparisons between reported reviews, whereas the latter is computed by comparing reported reviews with gold-standard reviews. The Pearson correlation coefficients for BG, BIG, and CG are, respectively, $-0.73$, $-0.79$, and $-0.72$. This result implies that there exists a strong positive correlation between honest reporting and accuracy in this task, a fact which is in agreement with our theoretical model.

\subsubsection{Completion Time}

The average completion time per group is shown in the second row of Table~\ref{tab:rewards_time}. We start by noting that Hypothesis 3 is not true. Surprisingly, the average time spent on the task by members of the group BG is statistically equivalent to the average time spent by members of the group CG since the null hypothesis cannot be rejected. 
The average completion time by members of the group BIG is the highest one amongst the three groups, and this result is statistically significant with $p\mbox{-value} \leq 0.05$. A possible explanation for this result is that reviewers work on the reviewing task more seriously by taking more time to complete it when they receive a brief explanation regarding some theoretical results of the proposed scoring method, whereas they could be quickly guessing how their peers would review the texts when the extra explanation about the theoretical results is not provided.

It is noteworthy that even though the average values might suggest that spending more time reviewing the texts results in higher bonuses and lower overall absolute errors, we do not find any significant correlation between these variables at an individual level.

\subsubsection{Consensus}

\begin{table}[t]
\caption{Accuracy of the average review (AVG) and the consensual review (CR) per group.
The average of the absolute difference between the aggregate evaluation scores and the corresponding gold-standard evaluation scores is shown per group and criteria. For each criterion, the lowest average absolute difference in each group is highlighted in bold. The standard deviations are in parenthesis. A total of 1000 bootstrap resamples were used.}
\label{tab:consensual_average}
\centering
\begin{tabular}{l l c c c c c c} 
\toprule
 & & \multicolumn{2}{c}{BG} &  \multicolumn{2}{c}{BIG} & \multicolumn{2}{c}{CG}\\     
     \cmidrule(r){3-4}   \cmidrule(r){5-6} \cmidrule(r){7-8}
 &    & AVG & CR & AVG & CR & AVG & CR \\
\midrule
       & Grammar   & \textbf{0.2622} & 0.2781   & 0.1238   & \textbf{0.1086} & \textbf{0.1285} & 0.1307 \\
       &           & (0.0920) & (0.1015)        & (0.0680) & (0.0631) & (0.0800) & (0.0829)    \\
Text 1 & Clarity   & 0.8219   & \textbf{0.8159} & 0.6211 & \textbf{0.6078} & 0.8633 & \textbf{0.8477} \\
       &           & (0.0903) & (0.0963)        & (0.0850) & (0.0987) & (0.1031) & (0.1156)  \\ 
       & Relevance & 0.2188   & \textbf{0.1462} & 0.2020 & \textbf{0.1382} & 0.2996 & \textbf{0.2156} \\
       &           & (0.0676) & (0.0536)        & (0.0622) & (0.0512) &  (0.0783) & (0.0696) \\ 
       \\
       & Grammar   & \textbf{0.1643} & 0.1667   & 0.0754 & \textbf{0.0698} & 0.0716 & \textbf{0.0679}  \\
       &           & (0.0819) & (0.0855)        & (0.0571) & (0.0541) & (0.545) & (0.0531)   \\ 
Text 2 & Clarity   & 0.4965   & \textbf{0.4314} & 0.3812 & \textbf{0.3036} & 0.5405 & 	\textbf{0.5022}  \\
       &           & (0.0879) & (0.0970)        & (0.0858) & (0.0860) & (0.0857) & (0.0987)   \\ 
       & Relevance & 0.3590   & \textbf{0.3508} & \textbf{0.4360} & 0.4957 & \textbf{0.4996 }& 0.5642\\
       &           & (0.0804) & (0.0933)        & (0.0948) & (0.1061) & (0.0909) & (0.0999) \\ 
       \\
       & Grammar   & 0.7598  & \textbf{0.6976}  & 0.7792 & \textbf{0.7198} & \textbf{1.0231 }& 1.0288\\
       &           & (0.1204) & (0.1473) &  (0.1222) &(0.1505) & (0.1161) & (0.1455)\\ 
Text 3 & Clarity   & 0.1379   & \textbf{0.0859} & \textbf{0.0000} & \textbf{0.0000} & 0.1615 & \textbf{0.1115}  \\
       &           & (0.0565) & (0.0400) &  (0.0000) & (0.0000) & (0.0522) & (0.0453) \\
       & Relevance & 0.1197   & \textbf{0.0699} & 0.1019 & \textbf{0.0597} & 0.2022 & \textbf{0.1320} \\ 
       &           & (0.0619) & (0.0400) & (0.0521) & (0.0335) & (0.0696) & (0.0537)\\ 
\bottomrule
\end{tabular}
\end{table}

Lastly, we tested the accuracy of the method proposed in Section 6 to find a consensual review. We compared the resulting consensual review with the average review by using a bootstrapping technique. For each group of reviewers, we randomly resampled with replacement reviews from the original dataset so as to obtain \emph{bootstrap resamples}. The size of each bootstrap resample was equal to the size of the original dataset, \textit{i.e.}, each bootstrap resample contained 50 data points (the original number of reviewers), each one consisting of 9 evaluation scores.

For each bootstrap resample, we aggregated evaluation scores individually using both the proposed method for finding a consensual review and the average method. For each evaluation score, the weight that each reviewer $i$ assigned to reviewer $j$'s reported evaluation score was computed according to equation (\ref{eq:consensus_weight}). The proper scoring rule $R$ and the constants $\gamma$ and $\lambda$ were set so as to reward agreement as in Section 4.3. Given that the best evaluation score in this task was $v=2$ and $\alpha = (1,1,1)$, each element of a reviewer's estimated posterior predictive distribution in (\ref{eq:est_post_pred_dist}) could take on only two values: $0.25$ and $0.5$. Consequently, the numerator in (\ref{eq:consensus_weight}) could take on only two values: $0.5$, if reviewer $i$ and $j$'s reported evaluation scores were the same, and $0.25$ otherwise.

After aggregating evaluation scores, we computed the accuracy of each aggregation method. The outcome measure was the absolute difference between each aggregate evaluation score and the corresponding gold standard score. Thus, the outcome measure was an integer with a value between zero and two, and the closer this value was to zero, the better the resulting accuracy. Table~\ref{tab:consensual_average} shows the average accuracy by group resulting from a total of 1000 bootstrap resamples.

Table~\ref{tab:consensual_average} shows that consensual evaluation scores are more accurate than average evaluation scores in 20 out of 27 cases, and equally accurate in one case. It comes as no surprise that consensual evaluation scores are more accurate in groups where review scores were provided since these groups reported more accurate reviews and their reported reviews were more similar, as previously discussed. We performed a statistical analysis to investigate whether or not these differences in accuracy are statistically significant. Since we used the same bootstrap resamples for both aggregation methods, the Wilcoxon signed-rank test was used. The null hypothesis was that the outcome measures resulting from both aggregation methods are equivalent. The alternative hypothesis was that the outcome measure resulting from the consensual method was less than the outcome measure resulting from the average method, which implies that the former is more accurate than the latter. All the resulting 27 $p$-values are extremely small $\left(< 10^{-10}\right)$. Therefore, we conclude that Hypothesis 4 is indeed true, \textit{i.e.}, the proposed method for finding a consensual review is, on average, more accurate than the average approach in this experiment. As discussed in Section 6, we believe this result happens because disparate reported reviews are less likely to have a big influence on the consensual review than on the average review.

\section{Conclusion}

We proposed a scoring method built on strictly proper scoring rules that induces honest reporting when outcomes are not observable. We illustrated the mechanics behind our scoring method by applying it to the peer-review process.
In order to do so, we modeled the peer-review process using a Bayesian model where the uncertainty regarding the quality of the manuscript is taken into account. The main assumptions in our model are that reviewers cannot be influenced by other reviewers, and reviewers are Bayesian decision-makers.

We then showed how our scoring method can be used to evaluate reported reviews and to encourage honest reporting by risk-neutral reviewers. The proposed method assigns scores based on how close reported reviews are, where closeness is defined by an underlying proper scoring rule. Under the aforementioned assumptions, we showed that risk-neutral reviewers strictly maximize their expected scores by honestly disclosing their reviews. We also proposed an extension of our model and scoring method to scenarios where reviewers evaluate a manuscript under several criteria. We discussed how honest reporting is related to accuracy in our model: when reviewers report honestly, the distribution of reported reviews convergences to the distribution that represents the quality of the manuscript as the number of reported reviews increases. 

Since all reviews are not always in agreement, we suggested an adaptation of the  method proposed by \cite{DeGroot:1974} to find a consensual review. Intuitively, the proposed method works as if the reviewers were going through several rounds of discussion, where in each round they are informed about others' reported reviews, and they update their own reviews using our predefined method in order to reach a consensus. Formally, each updated review is a convex combination of reported reviews, where review scores are used as part of the weights that reviewers assign to their peers' reviews. We showed that the resulting method always converges to a consensual review when reviewers' expected review scores are positive, and we provided behavioral foundations for the aggregation method.

We tested the efficacy of both the proposed scoring method and the proposed aggregation method on a peer-review experiment using Amazon's Mechanical Turk. Our experimental results corroborated the relationship between honest reporting and accuracy in our model. We empirically showed that providing review scores through pairwise comparisons results in more accurate reviews than the traditional peer-review process, where reviewers have no direct incentives for expressing their true reviews. Moreover, reviewers tended to agree more with each other when they received review scores. In addition, our method for finding a consensual review outperformed the traditional average method in our peer-review experiments.

For ease of exposition, our discussion on peer review was focused on scientific communication. However, our model and scoring method are readily applied to most peer-review settings, \textit{e.g.}, academic courses, clinical peer review, \textit{etc}. Moreover, our proposed method to incentivize honest reporting is readily applied to different domains. For example, our method can be used to incentivize honest feedback in reputation systems, where individuals rate a product/service after experiencing it, and to induce honest evaluation of different strategic plans in an organization's strategic planning process. The proposed aggregation method is also general in a sense that it can be applied to any decision analysis process where experts express their opinions through probability distributions over a set of exhaustive and mutually exclusive outcomes.

Given the positive results obtained in our peer-review experiments, an interesting open question is whether or not the methods proposed in this paper would perform as well in other domains, such as in the aforementioned reputation systems and strategic planning in organizations. Another question worth contemplating is whether or not incentives other than from the received scores play a role in our scoring method. For example, one can conjecture that altruism may play an important role in our scoring method. In our peer-review experiments, the performance of the reviewers not only affect their own review scores, but also the review scores of their peers. In other words, if reviewers do not put enough effort into reporting high-quality reviews, not only might they receive low review scores, but other reviews evaluated based on those erroneous reviews might also receive low review scores. Thus, an interesting future work is to investigate whether or not experts, in general, have an altruistic motive to put more effort into the underlying task in order to maximize the potential payoffs of their peers.

\section*{Acknowledgments}

The authors thank Carol Acton, Katherine Acheson, Stefan Rehm, Susan Gow, and Veronica Austen for providing gold-standard reviews for our experiments.

\section*{Appendix}

In this appendix, we describe the texts used in our experiments as well as the gold-standard reviews reported by five professors and tutors from the English and Literature Department at the University of Waterloo, henceforth referred to as the \emph{experts}.

\subsubsection*{Text 1}

An excerpt from the ``Sonnet XVII" by \cite{Neruda:2007}. Intentionally misspelled words are highlighted in bold.

\begin{quote}
``I do not love you as if you \textbf{was} salt-rose, or topaz,\\
or the \textbf{arrown} of carnations that spread fire:    \\
I love you as certain dark things are loved,  \\
secretly, between the \textbf{shadown} and the soul"      
\end{quote}

Table \ref{tab:text1} shows the evaluation scores reported by the experts. The gold-standard evaluation score for each criterion is the median/mode of the reported evaluation scores.

\begin{table} [H]
\caption{Evaluation scores reported by the experts for Text 1.}
\label{tab:text1}
\centering
\begin{tabular}{l c c c c c c} 
\toprule
 Criterion & Expert 1 & Expert 2 & Expert 3 &  Expert 4 &  Expert 5 & Median/Mode\\
\midrule
Grammar   & 1 & 0 & 1 & 0 & 1 & 1\\
Clarity   & 2 & 2 &	2 &	1 &	2 & 2\\
Relevance & 2 & 2 & 2 & 2 & 2 & 2\\ 
\bottomrule
\end{tabular}
\end{table}

\subsubsection*{Text 2}

An excerpt from ``The Cow" by \cite{Taylor:2010}. Intentionally misspelled words are highlighted in bold.

\begin{quote}
``THANK you, \textbf{prety} cow, that made \\
\textbf{Plesant} milk to soak my bread, \\
Every day and every night,  \\
Warm, and fresh, and sweet, and white." 
\end{quote}

Table \ref{tab:text2} shows the evaluation scores reported by the experts. The gold-standard evaluation score for each criterion is the median/mode of the reported evaluation scores.

\begin{table} [H]
\caption{Evaluation scores reported by the experts for Text 2.}
\label{tab:text2}
\centering
\begin{tabular}{l c c c c c c} 
\toprule
 Criterion & Expert 1 & Expert 2 & Expert 3 &  Expert 4 &  Expert 5 & Median/Mode\\
\midrule
Grammar   & 1 & 1 & 1 & 1 & 1 & 1\\
Clarity   & 2 & 2 &	2 &	1 &	2 & 2\\
Relevance & 1 & 0 & 0 & 1 & 1 & 1\\ 
\bottomrule
\end{tabular}
\end{table}

\subsubsection*{Text 3}

Random words in a semi-structured way. Each line starts with a noun followed by a verb in a wrong verb form. All the words in the same line start with a similar letter in order to mimic a poetic writing style.

\begin{quote}
``Baby bet binary boundaries bubbles \\
Carlos cease CIA conditionally curve \\
Daniel deny disease domino dumb \\
Faust fest fierce forced furbished"
\end{quote}

Table \ref{tab:text3} shows the evaluation scores reported by the experts. The gold-standard evaluation score for each criterion is the median/mode of the reported evaluation scores.

\begin{table} [H]
\caption{Evaluation scores reported by the experts for Text 3.}
\label{tab:text3}
\centering
\begin{tabular}{l c c c c c c} 
\toprule
 Criterion & Expert 1 & Expert 2 & Expert 3 &  Expert 4 &  Expert 5 & Median/Mode\\
\midrule
Grammar   & 0 & 1 & 0 & 0 & 0 & 0\\
Clarity   & 0 & 0 &	0 &	0 &	0 & 0\\
Relevance & 0 & 1 & 0 & 0 & 0 & 0\\ 
\bottomrule
\end{tabular}
\end{table}



\bibliographystyle{plainnat} 
\bibliography{bibliography} 

\begin{thebibliography}{56}
\providecommand{\natexlab}[1]{#1}
\providecommand{\url}[1]{\texttt{#1}}
\expandafter\ifx\csname urlstyle\endcsname\relax
  \providecommand{\doi}[1]{doi: #1}\else
  \providecommand{\doi}{doi: \begingroup \urlstyle{rm}\Url}\fi

\bibitem[{AICPA: American Institute of CPAs}(2012)]{AICPA:2012}
{AICPA: American Institute of CPAs}.
\newblock {Peer Review Program Manual}.
\newblock 2012.
\newblock Retrieved from \url{http://www.aicpa.org/InterestAreas/PeerReview}.

\bibitem[Argenti(1968)]{Argenti:1968}
J.~Argenti.
\newblock \emph{{Corporate Planning: a Practical Guide}}.
\newblock Number~2. Routledge, 1968.

\bibitem[Bacon et~al.(2012)Bacon, Chen, Kash, Parkes, Rao, and
  Sridharan]{Bacon:2012}
D.~F. Bacon, Y.~Chen, I.~Kash, D.~C. Parkes, M.~Rao, and M.~Sridharan.
\newblock {Predicting Your Own Effort}.
\newblock In \emph{Proceedings of the 11th International Conference on
  Autonomous Agents and Multiagent Systems}, pages 695--702, 2012.

\bibitem[Bernardo and Smith(1994)]{Bernardo:1994}
J.~M. Bernardo and A.~F.~M. Smith.
\newblock \emph{{Bayesian Theory}}.
\newblock John Wiley \& Sons, 1994.

\bibitem[Bornmann et~al.(2007)Bornmann, Mutz, and Daniel]{Bornmann:2007}
L.~Bornmann, R.~Mutz, and H.~D. Daniel.
\newblock {Gender Differences in Grant Peer Review: A Meta-Analysis}.
\newblock \emph{Journal of Informetrics}, 1\penalty0 (3):\penalty0 226 -- 238,
  2007.

\bibitem[Budden et~al.(2008)Budden, Tregenza, Aarssen, Koricheva, Leimu, and
  Lortie]{Budden:2008}
A.~E. Budden, T.~Tregenza, L.~W. Aarssen, J.~Koricheva, R.~Leimu, and C.~J.
  Lortie.
\newblock {Double-Blind Review Favours Increased Tepresentation of Female
  Authors}.
\newblock \emph{Trends in Ecology and Evolution}, 23\penalty0 (1):\penalty0
  4--6, 2008.

\bibitem[Buhrmester et~al.(2011)Buhrmester, Kwang, and
  Gosling]{Buhrmester:2011}
M.~D. Buhrmester, T.~Kwang, and S.~D. Gosling.
\newblock {Amazon's Mechanical Turk: A New Source of Inexpensive, Yet
  High-Quality, Data?}
\newblock \emph{Perspectives on Psychological Science}, 6\penalty0
  (1):\penalty0 3--5, 2011.

\bibitem[Carvalho and Larson(2010)]{Carvalho:2010}
A.~Carvalho and K.~Larson.
\newblock {Sharing a Reward Based on Peer Evaluations}.
\newblock In \emph{Proceedings of the 9th International Conference on
  Autonomous Agents and Multiagent Systems}, pages 1455--1456, 2010.

\bibitem[Carvalho and Larson(2011)]{Carvalho:2011}
A.~Carvalho and K.~Larson.
\newblock {A Truth Serum for Sharing Rewards}.
\newblock In \emph{Proceedings of the 10th International Conference on
  Autonomous Agents and Multiagent Systems}, pages 635--642, 2011.

\bibitem[Carvalho and Larson(2012)]{Carvalho:2012}
A.~Carvalho and K.~Larson.
\newblock {Sharing Rewards Among Strangers Based on Peer Evaluations}.
\newblock \emph{Decision Analysis}, 9\penalty0 (3):\penalty0 253--273, 2012.

\bibitem[Carvalho and Larson(2013)]{Carvalho:2013}
A.~Carvalho and K.~Larson.
\newblock {A Consensual Linear Opinion Pool}.
\newblock In \emph{Proceedings of the 23rd International Joint Conference on
  Artificial Intelligence}, pages 2518--2524, 2013.

\bibitem[Clemen and Winkler(1999)]{Clemen:1999}
R.~T. Clemen and R.~L. Winkler.
\newblock {Combining Probability Distributions From Experts in Risk Analysis}.
\newblock \emph{Risk Analysis}, 19:\penalty0 187--203, 1999.

\bibitem[Dans(1993)]{Dans:1993}
P.~E. Dans.
\newblock {Clinical Peer Review: Burnishing a Tarnished Icon}.
\newblock \emph{Annals of Internal Medicine}, 118\penalty0 (7):\penalty0
  566--568, 1993.

\bibitem[DeGroot(1974)]{DeGroot:1974}
M.~H. DeGroot.
\newblock {Reaching a Consensus}.
\newblock \emph{Journal of the American Statistical Association}, 69\penalty0
  (345):\penalty0 118--121, 1974.

\bibitem[Epstein(1969)]{Epstein:1969}
E.~S. Epstein.
\newblock {A Scoring System for Probability Forecasts of Ranked Categories}.
\newblock \emph{Journal of Applied Meteorology}, 8\penalty0 (6):\penalty0
  985--987, 1969.

\bibitem[Falagas et~al.(2006)Falagas, Zouglakis, and Kavvadia]{Falagas:2006}
M.~E. Falagas, G.~M. Zouglakis, and P.~K. Kavvadia.
\newblock {How Masked Is the “Masked Peer Review” of Abstracts Submitted to
  International Medical?}
\newblock \emph{Mayo Clinic Proceedings}, 81\penalty0 (5):\penalty0 705, 2006.

\bibitem[Friedman(1983)]{Friedman:1983}
D.~Friedman.
\newblock {Effective Scoring Rules for Probabilistic Forecasts}.
\newblock \emph{Management Science}, 29\penalty0 (4):\penalty0 447--454, 1983.

\bibitem[Gneiting and Raftery(2007)]{Gneiting:2007}
T.~Gneiting and A.~E. Raftery.
\newblock {Strictly Proper Scoring Rules, Prediction, and Estimation}.
\newblock \emph{Journal of the American Statistical Association}, 102\penalty0
  (477):\penalty0 359--378, 2007.

\bibitem[Godlee et~al.(2003)Godlee, Jefferson, Callaham, Clarke, Altman,
  Bastian, Bingham, and Deeks]{Godlee:2003}
F.~Godlee, T.~Jefferson, M.~Callaham, J.~Clarke, D.~Altman, H.~Bastian,
  C.~Bingham, and J.~Deeks.
\newblock \emph{{Peer Review in Health Sciences}}.
\newblock BMJ books London, 2003.

\bibitem[Hanson(2003)]{Hanson:2003}
R.~Hanson.
\newblock {Combinatorial Information Market Design}.
\newblock \emph{Information Systems Frontiers}, 5\penalty0 (1):\penalty0
  107--119, 2003.

\bibitem[Horton et~al.(2011)Horton, Rand, and Zeckhauser]{Horton:2011}
J.~J. Horton, D.~G. Rand, and R.~J. Zeckhauser.
\newblock {The Online Laboratory: Conducting Experiments in a Real Labor
  Market}.
\newblock \emph{Experimental Economics}, 14\penalty0 (3):\penalty0 399--425,
  2011.

\bibitem[Huang and Fu(2013)]{Huang:2013}
S.-W. Huang and W.-T. Fu.
\newblock {Enhancing Reliability Using Peer Consistency Evaluation in Human
  Computation}.
\newblock In \emph{Proceedings of the 2013 Conference on Computer Supported
  Cooperative Work}, pages 639--648, 2013.

\bibitem[Ipeirotis(2010)]{Ipeirotis:2010}
P.~G. Ipeirotis.
\newblock {Analyzing the Amazon Mechanical Turk Marketplace}.
\newblock \emph{XRDS Crossroads: The ACM Magazine for Students}, 17\penalty0
  (2):\penalty0 16--21, 2010.

\bibitem[Jurca and Faltings(2009)]{Jurca:2009}
R.~Jurca and B.~Faltings.
\newblock {Mechanisms for Making Crowds Truthful}.
\newblock \emph{Journal of Artificial Intelligence Research}, 34:\penalty0
  209--253, 2009.

\bibitem[Justice et~al.(1998)Justice, Cho, Winker, Berlin, Rennie, and {The
  PEER Investigators}]{Justice:1998}
A.~C. Justice, M.~K. Cho, M.~A. Winker, J.~A. Berlin, D.~Rennie, and {The PEER
  Investigators}.
\newblock {Does Masking Author Identity Improve Peer Review Quality?: A
  Randomized Controlled Trial}.
\newblock \emph{Journal of the American Medical Association}, 280\penalty0
  (3):\penalty0 240--242, 1998.

\bibitem[Lock(1985)]{Lock:1985}
S.~Lock.
\newblock \emph{{A Difficult Balance: Editorial Peer Review in Medicine}}.
\newblock Nuffield Provincial Hospitals Trust, 1985.

\bibitem[{LSC: Legal Services Commission}(2005)]{LSC:2005}
{LSC: Legal Services Commission}.
\newblock {Independent Peer Review}.
\newblock 2005.
\newblock Retrieved from
  \url{http://www.legalservices.gov.uk/civil/how/mq_peerreview.asp}.

\bibitem[Marge et~al.(2010)Marge, Banerjee, and Rudnicky]{Marge:2010}
M.~Marge, S.~Banerjee, and A.~I. Rudnicky.
\newblock {Using the Amazon Mechanical Turk for Transcription of Spoken
  Language}.
\newblock In \emph{Proceedings of the 2010 IEEE International Conference on
  Acoustics Speech and Signal Processing}, pages 5270--5273, 2010.

\bibitem[Mason and Suri(2012)]{Mason:2012}
W.~Mason and S.~Suri.
\newblock {Conducting Behavioral Research on Amazon's Mechanical Turk}.
\newblock \emph{Behavior Research Methods}, 44\penalty0 (1):\penalty0 1--23,
  2012.

\bibitem[Miller et~al.(2005)Miller, Resnick, and Zeckhauser]{Miller:2005}
N.~Miller, P.~Resnick, and R.~Zeckhauser.
\newblock {Eliciting Informative Feedback: The Peer-Prediction Method}.
\newblock \emph{Management Science}, 51\penalty0 (9):\penalty0 1359--1373,
  2005.

\bibitem[Murphy(1970)]{Murphy:1971}
A.~H. Murphy.
\newblock {A Note on the Ranked Probability Score}.
\newblock \emph{Journal of Applied Meteorology}, 10\penalty0 (1):\penalty0
  155--156, 1970.

\bibitem[Nakazono(2013)]{Nakazono:2013}
Y.~Nakazono.
\newblock {Strategic Behavior of Federal Open Market Committee Board Members:
  Evidence from Members' Forecasts}.
\newblock \emph{Journal of Economic Behavior \& Organization}, 93:\penalty0
  62--70, 2013.

\bibitem[Nau(1985)]{Nau:1985}
R.~F. Nau.
\newblock {Should Scoring Rules Be ``Effective"?}
\newblock \emph{Management Science}, 31\penalty0 (5):\penalty0 527--535, 1985.

\bibitem[Navarro et~al.(2006)Navarro, Griffiths, Steyvers, and
  Lee]{Navarro:2006}
D.~J. Navarro, T.~L. Griffiths, M.~Steyvers, and M.~D. Lee.
\newblock {Modeling Individual Differences with Dirichlet Processes}.
\newblock \emph{Journal of Mathematical Psychology}, 50\penalty0 (2):\penalty0
  101--122, 2006.

\bibitem[Neruda(2007)]{Neruda:2007}
P.~Neruda.
\newblock \emph{100 Love Sonnets}.
\newblock Exile, {Bilingual} edition, 2007.

\bibitem[Newcombe and Bouton(2009)]{Newcombe:2009}
N.~S. Newcombe and M.~E. Bouton.
\newblock {Masked Reviews Are Not Fairer Reviews}.
\newblock \emph{Perspectives on Psychological Science}, 4\penalty0
  (1):\penalty0 62--64, 2009.

\bibitem[Pappano(2012)]{Pappano:2012}
L.~Pappano.
\newblock {The Year of the MOOC}.
\newblock \emph{New York Times}, page ED26, November 4th, 2012.

\bibitem[Plous(1993)]{Plous:1993}
S.~Plous.
\newblock \emph{{The Psychology of Judgment and Decision Making}}.
\newblock Mcgraw-Hill Book Company, 1993.

\bibitem[Prelec(2004)]{Prelec:2004}
D.~Prelec.
\newblock {A Bayesian Truth Serum for Subjective Data}.
\newblock \emph{Science}, 306\penalty0 (5695):\penalty0 462--466, 2004.

\bibitem[Primack and Marrs(2008)]{Primack:2008}
R.~B. Primack and R.~Marrs.
\newblock {Bias in the Review Process}.
\newblock \emph{Biological Conservation}, 141\penalty0 (12):\penalty0
  2919--2920, 2008.

\bibitem[Radanovic and Faltings(2013)]{Radanovic:2013}
G.~Radanovic and B.~Faltings.
\newblock {A Robust Bayesian Truth Serum for Non-Binary Signals}.
\newblock In \emph{Proceedings of the 27th AAAI Conference on Artificial
  Intelligence}, 2013.

\bibitem[Robinson(2001)]{robinson:2001}
R.~Robinson.
\newblock {Calibrated Peer Review}.
\newblock \emph{The American Biology Teacher}, 63\penalty0 (7):\penalty0
  474--480, 2001.

\bibitem[Roos et~al.(2011)Roos, Rothe, and Scheuermann]{Roos:2011}
M.~Roos, J.~Rothe, and B.~Scheuermann.
\newblock {How to Calibrate the Scores of Biased Reviewers by Quadratic
  Programming}.
\newblock In \emph{Proceedings of the Twenty-Fifth Conference on Artificial
  Intelligence}, pages 255--260, 2011.

\bibitem[Ross(1995)]{Ross:1995}
S.~M. Ross.
\newblock \emph{Stochastic Processes (Wiley Series in Probability and
  Statistics)}.
\newblock Wiley, 2 edition, 1995.

\bibitem[Shaw et~al.(2011)Shaw, Horton, and Chen]{Shaw:2011}
A.~D. Shaw, J.~J. Horton, and D.~L. Chen.
\newblock {Designing Incentives for Inexpert Human Raters}.
\newblock In \emph{Proceedings of the ACM 2011 Conference on Computer Supported
  Cooperative Work}, pages 275--284, 2011.

\bibitem[Snow et~al.(2008)Snow, O'Connor, Jurafsky, and Ng]{Snow:2008}
R.~Snow, B.~O'Connor, D.~Jurafsky, and A.~Y. Ng.
\newblock {Cheap and Fast---But is it Good? Evaluating Non-Expert Annotations
  for Natural Language Tasks}.
\newblock In \emph{Proceedings of the Conference on Empirical Methods in
  Natural Language Processing}, pages 254--263, 2008.

\bibitem[{Sta{\"e}l von Holstein}(1970)]{Holstein:1970}
C.-A.~S. {Sta{\"e}l von Holstein}.
\newblock {A Family of Strictly Proper Scoring Rules Which Are Sensitive to
  Distance}.
\newblock \emph{Journal of Applied Meteorology}, 9\penalty0 (3):\penalty0
  360--364, 1970.

\bibitem[Taylor et~al.(2010)Taylor, Taylor, and Greenaway]{Taylor:2010}
J.~Taylor, A.~Taylor, and K.~Greenaway.
\newblock \emph{Little Ann and Other Poems}.
\newblock Nabu Press, 2010.

\bibitem[Tversky and Kahneman(1974)]{Tversky:1974}
A.~Tversky and D.~Kahneman.
\newblock {Judgment under Uncertainty: Heuristics and Biases}.
\newblock \emph{Science}, 185\penalty0 (4157):\penalty0 1124--1131, 1974.

\bibitem[van Rooyen(2001)]{vanRooyen:2001}
S.~van Rooyen.
\newblock {The Evaluation of Peer-Review Quality}.
\newblock \emph{Learned Publishing}, 14\penalty0 (2):\penalty0 85--91, 2001.

\bibitem[van Rooyen et~al.(1999)van Rooyen, Black, and Godlee]{vanRooyen:1999}
S.~van Rooyen, N.~Black, and F.~Godlee.
\newblock {Development of the Review Quality Instrument {(RQI)} for Assessing
  Peer Reviews of Manuscripts}.
\newblock \emph{Journal of Clinical Epidemiology}, 52\penalty0 (7):\penalty0
  625 -- 629, 1999.

\bibitem[Weiss(2009)]{Weiss:2009}
R.~R.~J. Weiss.
\newblock \emph{{Optimally Aggregating Elicited Expertise: A Proposed
  Application of the Bayesian Truth Serum for Policy Analysis}}.
\newblock PhD thesis, Massachusetts Institute of Technology, 2009.

\bibitem[Wenneras and Wold(1997)]{Wenneras:1997}
C.~Wenneras and A.~Wold.
\newblock {Nepotism and Sexism in Peer-Review}.
\newblock \emph{Nature}, 387:\penalty0 341--343, 1997.

\bibitem[Winkler and Murphy(1968)]{Winkler:1968}
R.~L. Winkler and A.~H. Murphy.
\newblock {``Good" Probability Assessors}.
\newblock \emph{Journal of Applied Meteorology}, 7\penalty0 (5):\penalty0
  751--758, 1968.

\bibitem[Witkowski and Parkes(2012)]{Witkowski:2012}
J.~Witkowski and D.~C. Parkes.
\newblock {A Robust Bayesian Truth Serum for Small Populations}.
\newblock In \emph{Proceedings of the 26th AAAI Conference on Artificial
  Intelligence}, 2012.

\bibitem[Yankauer(1991)]{Yankauer:1991}
A.~Yankauer.
\newblock {How Blind is Blind Review?}
\newblock \emph{American Journal of Public Health}, 81\penalty0 (7):\penalty0
  843--845, 1991.

\end{thebibliography}



\end{document}